\documentclass{IEEEtran}
\usepackage{cite}
\usepackage{amsmath,amssymb,amsfonts}
\usepackage{subfigure}
\usepackage{algorithmic}
\usepackage{graphicx}
\usepackage{textcomp}
\usepackage{caption}
\usepackage{color}
\usepackage{hhline}
\usepackage{tabu,multirow}
\usepackage{makecell}
\usepackage{epstopdf}
\usepackage[margin=2cm]{geometry}
\usepackage{indentfirst}
\usepackage[T1]{fontenc}
\usepackage{float}
\usepackage{subeqnarray}
\usepackage{cases}
\usepackage{times}
\usepackage{booktabs}
\usepackage{multirow}
\usepackage{siunitx}

\input{tcilatex}
\begin{document}

\title{Sliding Mode Attitude Maneuver Control for Rigid Spacecraft without
Unwinding}
\author{Rui-Qi Dong, \IEEEmembership{Student Member, IEEE}, Ai-Guo Wu, %
\IEEEmembership{Member, IEEE} and Ying Zhang,
\thanks{%
This paragraph of the first footnote will contain the date on which you
submitted your brief for review. It will also contain support information,
including sponsor and financial support acknowledgment. For example, ``This
work was supported in part by the U.S. Department of Commerce under Grant
BS123456.'' } \thanks{%
The next few paragraphs should contain the authors' current affiliations,
including current address and e-mail. For example, F. A. Author is with the
National Institute of Standards and Technology, Boulder, CO 80305 USA
(e-mail: author@boulder.nist.gov). } \thanks{%
S. B. Author, Jr., was with Rice University, Houston, TX 77005 USA. He is
now with the Department of Physics, Colorado State University, Fort Collins,
CO 80523 USA (e-mail: author@lamar.colostate.edu).} \thanks{%
T. C. Author is with the Electrical Engineering Department, University of
Colorado, Boulder, CO 80309 USA, on leave from the National Research
Institute for Metals, Tsukuba, Japan (e-mail: author@nrim.go.jp).}}
\maketitle

\begin{abstract}
In this paper, attitude maneuver control without unwinding phenomenon is
investigated for rigid spacecraft. First, a novel switching function is
constructed by a hyperbolic sine function. It is shown that the spacecraft
system possesses the unwinding-free performance when the system states are
on the sliding surface. Based on the designed switching function, a sliding
mode controller is developed to ensure the robustness of the attitude
maneuver control system. Another essential feature of the presented attitude
control law is that a dynamic parameter is introduced to guarantee the
unwinding-free performance when the system states are outside the sliding
surface. The simulation results demonstrate that the unwinding phenomenon is
avoided during the attitude maneuver of a rigid spacecraft by adopting the
constructed switching function and the proposed attitude control scheme.

\end{abstract}


\section{Introduction}

Attitude maneuver control of rigid spacecraft has gained a great deal of
attention in the last decades due to the benefits attained through its wide
applications such as satellite communication, ocean surveillance, and
spacecraft pointing~\cite{hughes2012spacecraft}.
Many control strategies have been proposed for solving the attitude maneuver
control problem, such as optimal control~\cite{Sharma2004Optimal}, event
trigger control~\cite{amrr2019event}, linear parameter varying control~\cite%
{jin2018lpv}, model predictive control~\cite{bayat2018model},
backstepping control~\cite{ji2018vibration}, and so on. However, the
attitude controller design is still challenging
due to two aspects: the inherent nonlinearity of the spacecraft attitude
dynamics and the unwinding phenomenon during spacecraft attitude maneuver.

Sliding Mode Control (SMC)
has been widely applied to deal with the nonlinearity of the spacecraft
attitude dynamics due to its strong robustness to disturbance~\cite%
{chen1993sliding,tiwari2016attitude,wu2018adaptive,zou2011finite,xu2015study}%
. By using the integral sliding mode, a high-order sliding mode controller
was proposed in~\cite{tiwari2016attitude} to address the chattering issue of
SMC methods. In~\cite{wu2018adaptive}, an adaptive law was proposed to
estimate the upper bound of the unknown lumped disturbances, including
external disturbance, flexible vibration, and inertia uncertainty. In~\cite%
{zou2011finite}, a finite time controller was presented for attitude
synchronization. In order to resolve the singular problem of the traditional
terminal and faster terminal sliding-mode control designs, a nonsingular
finite-time control approach was developed in~\cite{xu2015study}.

In the above control schemes, the unit quaternion was adopted to describe
spacecraft attitude. However, the quaternion has four parameters, which can
result in an extra constraint because three parameters are enough to
describe the spacecraft attitude%
~\cite{chaturvedi2011rigid}. Thus, the Modified Rodrigues Parameter (MRP)
was adopted to represent the spacecraft attitude, and plenty of controllers
were proposed~\cite%
{du2011finite,tsiotras1998further,cong2013backstepping,crassidis1996sliding}%
. In~\cite{du2011finite}, two finite-time attitude control laws were
developed for single and multiple spacecraft, respectively.
An attitude stabilization control law without angular velocity measurements
was established for rigid spacecraft in~\cite{tsiotras1998further}. The
attitude maneuver control was investigated in~\cite{cong2013backstepping},
and a backstepping based adaptive sliding mode control strategy was
proposed. In~\cite{crassidis1996sliding}, a sliding mode control method was
presented to solve the attitude tracking problem of rigid spacecraft.

Note that the unwinding problem was neglected by the aforementioned
researches when the MRP was adopted to represent the spacecraft attitude. A
typical feature of the MRP is its non-uniqueness, that is, a specific
spacecraft orientation can be represented by two different MRP vectors.
These two MRP vectors correspond to two different rotation angles and
opposite rotation directions about the same Euler axis. The sum of these two
rotation angles is $2\pi$. The unwinding phenomenon is that the rotation
angle larger than $\pi$ is performed by an attitude maneuver controller,
which results in extra control effort. Nevertheless, to the best knowledge
of the author, there is no unwinding-free result about the attitude maneuver
control for the rigid spacecraft based on the MRP description.

Based on the above discussions, we aim to design an unwinding-free sliding
mode attitude maneuver control law for a rigid spacecraft based on MRP
representation. First of all, a switching function is designed using a
hyperbolic sine function. Rigorous proof about the unwinding-free
performance of the spacecraft system when the system states are on the
sliding surface is given. Secondly, a sliding mode control law is presented
to guarantee that all the system states arrive at the constructed sliding
surface. Furthermore, it is proven that the unwinding-free performance of
the proposed controller with a dynamic parameter is achieved before the
system states reach the sliding surface. Finally, the simulations are
performed to show the unwinding-free performance of the proposed controller.
Compared with the SMC control law in~\cite{crassidis1996sliding}, the
proposed control law in this paper possesses faster convergence rate, and
smaller control torque.

Throughout this paper, we use the italic-font notation for a scalar variable
(as $\rho$), the bold-font notation for a vector (as $\boldsymbol{\sigma}$),
and the capital-letter notation for a matrix (as $M$). The set of $n$%
-dimensional real vectors, and the set of $m$-by-$n$ real matrices, are
denoted by $\mathbb{R}^{n}$ and $\mathbb{R}^{m\times n}$, respectively. In
addition, $\boldsymbol{0}$ and $I_{3}$ respectively denote a $3$-dimensional
zero vector and a $3\times 3$ identity matrix. We use $\left\Vert\cdot\right%
\Vert$ to represent the $2$-norm of a vector, and $\otimes$ to represent the
MRP multiplication. Moreover, the following two hyperbolic functions are
used, $\cosh x=\frac{\boldsymbol{e}^{x}+\boldsymbol{e}^{-x}}{2}$ and $\sinh
x=\frac{\boldsymbol{e}^{x}-\boldsymbol{e}^{-x}}{2}$ for $x\in \mathbb{R}$.
Moreover, the following relations are used: $\frac{\mathrm{d}\left( \cosh
x\right) }{\mathrm{d}x}=\sinh x$, $\frac{\mathrm{d}\left( \sinh x\right) }{%
\mathrm{d}x}=\cosh x$, $\frac{\mathrm{d}\left( \tan x\right) }{\mathrm{d}x}=%
\frac{1}{x^2+1}$, 
and $\cos^{2} x=\frac{1}{\tan^{2} x+1}$.

\section{Mathematic model}
Before given the mathematical model, we first give some denotations. For any vector $\boldsymbol{x}=\left[ x_{1}\
x_{2}\ x_{3}\right] ^{\mathrm{T}}\in \mathbb{R}^{3}$, let
\begin{equation*}
\boldsymbol{x}^{\times }=\left[
\begin{array}{ccc}
0 & -x_{3} & x_{2} \\
x_{3} & 0 & -x_{1} \\
-x_{2} & x_{1} & 0%
\end{array}%
\right],
\end{equation*}
and
\begin{equation}
M\left( \boldsymbol{x}\right) =\frac{\left( 1-\left\Vert \boldsymbol{x}%
\right\Vert ^{2}\right) I_{3}+2\boldsymbol{x}^{\times }+2\boldsymbol{xx}^{%
\mathrm{T}}}{4}.  \label{M}
\end{equation}%

\subsection{Attitude kinematics and dynamics}

By using the Modified Rodrigues Parameter (briefly, MRP), the rigid
spacecraft attitude dynamics can be given as~\cite{amrr2019event}
\begin{equation}
\left\{
\begin{array}{l}
\dot{\boldsymbol{\sigma }}=M\left( \boldsymbol{\sigma }\right) \boldsymbol{%
\omega }, \\
J\boldsymbol{\dot{\omega}}=-\boldsymbol{\omega }^{\times }J\boldsymbol{%
\omega }+\boldsymbol{u}+\boldsymbol{d},%
\end{array}%
\right.  \label{dynamics}
\end{equation}%
where $\boldsymbol{\sigma }\in \mathbb{R}^{3}$ is the spacecraft attitude of
the body frame $\mathcal{F}_{\mathrm{b}}$ with respect to the inertia frame $%
\mathcal{F}_{\mathrm{I}}$, $\boldsymbol{\omega }\in \mathbb{R}^{3}$ is the
angular velocity expressed in $\mathcal{F}_{\mathrm{b}}$; $J\in \mathbb{R}%
^{3\times 3}$ is the inertia matrix, $\boldsymbol{u}\in \mathbb{R}^{3}$ is
the control input, and $\boldsymbol{d}\in \mathbb{R}^{3}$ is the
disturbance.
\subsection{Attitude error kinematics and dynamics}

Let $\boldsymbol{\sigma }_{\mathrm{d}}\in \mathbb{R}^{3}$ denotes the
spacecraft attitude of the desired frame $\mathcal{F}_{\mathrm{d}}$ with
respect to the inertia frame $\mathcal{F}_{\mathrm{I}}$, and $\boldsymbol{%
\omega }_{\mathrm{d}}\in \mathbb{R}^{3}$ denotes the attitude angular
velocity expressed in the desired frame $\mathcal{F}_{\mathrm{d}}$. In
addition, denote $\boldsymbol{\sigma }_{\mathrm{e}}\in \mathbb{R}^{3}$ as
the attitude error between the desired attitude $\boldsymbol{\sigma }_{%
\mathrm{d}}$ and the body attitude $\boldsymbol{\sigma }_{\mathrm{b}}$.
Then, the attitude error $\boldsymbol{\sigma }_{\mathrm{e}}$ can be given by%
\begin{align}
\boldsymbol{\sigma }_{\mathrm{e}}&=\boldsymbol{\sigma \otimes \sigma }_{%
\mathrm{d}}^{\ast },  \notag \\
&=\frac{\left( 1-\left\Vert \boldsymbol{\sigma }\right\Vert ^{2}\right)
\boldsymbol{\sigma }_{\mathrm{d}}+\left( 1-\left\Vert \boldsymbol{\sigma }_{%
\mathrm{d}}\right\Vert ^{2}\right) \boldsymbol{\sigma }+2\boldsymbol{\sigma }%
_{\mathrm{d}}^{\times }\boldsymbol{\sigma }}{1+\left\Vert \boldsymbol{\sigma
}_{\mathrm{d}}\right\Vert ^{2}\left\Vert \boldsymbol{\sigma }\right\Vert
^{2}+2\boldsymbol{\sigma }_{\mathrm{d}}^{\mathrm{T}}\boldsymbol{\sigma }},
\label{sigma_e}
\end{align}%
where $\boldsymbol{\sigma}_{\mathrm{d}}^{\ast }=-\boldsymbol{\sigma}_{%
\mathrm{d}}$. Denote $\boldsymbol{\omega }_{\mathrm{e}}\in \mathbb{R}^{3}$
as the attitude angular velocity error between $\boldsymbol{\omega }_{%
\mathrm{d}}$ and $\boldsymbol{\omega }_{\mathrm{b}}$. Then, we have
\begin{align}
\boldsymbol{\omega }_{\mathrm{e}}&=\boldsymbol{\omega }-R\left( \boldsymbol{%
\sigma }_{\mathrm{e}}\right) \boldsymbol{\omega }_{\mathrm{d}},
\label{omega_e}
\end{align}%
where $R\left( \boldsymbol{\sigma }_{\mathrm{e}}\right) $ is the rotation
matrix from the desired frame $\mathcal{F}_{\mathrm{d}}$ to the body frame $%
\mathcal{F}_{\mathrm{b}}$, and can be expressed as%
\begin{equation*}
R\left( \boldsymbol{\sigma }_{\mathrm{e}}\right) =I_{3}+\frac{8\boldsymbol{%
\sigma }_{\mathrm{e}}^{\times }\boldsymbol{\sigma }_{\mathrm{e}}^{\times
}-4\left( 1-\left\Vert \boldsymbol{\sigma }_{\mathrm{e}}\right\Vert
^{2}\right) \boldsymbol{\sigma }_{\mathrm{e}}^{\times }}{\left( 1+\left\Vert
\boldsymbol{\sigma }_{\mathrm{e}}\right\Vert ^{2}\right) ^{2}}.
\end{equation*}%
Thus, the attitude error kinematics can be obtained as~\cite{amrr2019event}
\begin{align}
\dot{\boldsymbol{\sigma }}_{\mathrm{e}}=M\left( \boldsymbol{\sigma }_{%
\mathrm{e}}\right) \boldsymbol{\omega }_{\mathrm{e}}.
\label{error kinematics}
\end{align}
In addition, the rotation matrix $R\left( \boldsymbol{\sigma }_{\mathrm{e}%
}\right)$ satisfies $\dot{R}\left( \boldsymbol{\sigma }_{\mathrm{e}}\right)=-%
\boldsymbol{\omega} _{\mathrm{e}}^{\times }R\left( \boldsymbol{\sigma }_{%
\mathrm{e}}\right)$. It follows from~(\ref{omega_e}) that%
\begin{equation}
\boldsymbol{\dot{\omega}}_{\mathrm{e}}=\boldsymbol{\dot{\omega}}-R\left(
\boldsymbol{\sigma }_{\mathrm{e}}\right) \boldsymbol{\dot{\omega}}_{\mathrm{d%
}}+\boldsymbol{\omega} _{\mathrm{e}}^{\times }R\left( \boldsymbol{\sigma }_{%
\mathrm{e}}\right)\boldsymbol{\omega }_{\mathrm{d}}.  \label{dotomega_e}
\end{equation}

For a rest-to-rest attitude maneuver control problem, the desired angular
velocity $\boldsymbol{\omega}_{\mathrm{d}}$ satisfies $\boldsymbol{\omega}_{%
\mathrm{d}}=\boldsymbol{0}$ and $\dot{\boldsymbol{\omega}}_{\mathrm{d}}=%
\boldsymbol{0}$. Thus, it can be obtained from~(\ref{dotomega_e}) that $\dot{%
\boldsymbol{\omega}}_{\mathrm{e}}=\dot{\boldsymbol{\omega}}$ holds. By
substituting this relation into the second equation of~({\ref{dynamics}}),
and using~(\ref{error kinematics}), the following rest-to-rest attitude
maneuver error dynamics for a rigid spacecraft based on MRP can be obtained~%
\cite{cong2013backstepping},
\begin{equation}
\left\{
\begin{array}{l}
\dot{\boldsymbol{\sigma }}_{\mathrm{e}}=M\left( \boldsymbol{\sigma }_{%
\mathrm{e}}\right) \boldsymbol{\omega }_{\mathrm{e},} \\
J\boldsymbol{\dot{\omega}}_{\mathrm{e}}=-\boldsymbol{\omega }_{\mathrm{e}%
}^{\times }J\boldsymbol{\omega }_{\mathrm{e}}+\boldsymbol{u}+\boldsymbol{d},%
\end{array}%
\right.  \label{systemmodel}
\end{equation}%
where the matrix $M\left( \boldsymbol{\sigma }_{\mathrm{e}}\right) $ in
terms of $\boldsymbol{\sigma }_{\mathrm{e}}$ can be obtained by replacing $%
\boldsymbol{x}$ with $\boldsymbol{\sigma }_{\mathrm{e}}$ in~(\ref{M}).

In addition, according to the Euler's principle rotation theorem~\cite%
{shuster1993survey}, the rest-to-rest attitude maneuver of a rigid
spacecraft can also be described as that the spacecraft performs a rotation
from the body frame $\mathcal{F}_{\mathrm{b}}$ to the desired frame $%
\mathcal{F}_{\mathrm{b}}$ about a certain Euler axis, which is a unit
vector. Suppose that the rotation angle and the Euler axis of this rotation
are denoted by $\theta \left(t\right)\in \mathbb{R}$ and $\boldsymbol{e}\in
\mathbb{R}^{3}$, respectively. Then, the rest-to-rest attitude maneuver
error $\boldsymbol{\sigma}_{\mathrm{e}}$ from $\mathcal{F}_{\mathrm{b}}$ to $%
\mathcal{F}_{\mathrm{d}}$ can be expressed as
\begin{equation}
\boldsymbol{\sigma }_{\mathrm{e}}=\boldsymbol{e}\tan \frac{%
\theta\left(t\right) }{4}.\   \label{sigma}
\end{equation}
According to~(\ref{sigma_e}), $\boldsymbol{\sigma }_{\mathrm{e}%
}\left(0\right)$ can be obtained as long as the initial attitude $%
\boldsymbol{\sigma}\left(0\right)$ of $\boldsymbol{\sigma}$ and the desired
attitude $\boldsymbol{\sigma}_{\mathrm{d}}$ are given. Thus, the following
relations can be obtained by~(\ref{sigma}),
\begin{equation}
\theta\left(t\right)=4\arctan \boldsymbol{e}^{\mathrm{T}}\boldsymbol{\sigma}%
_{\mathrm{e}},  \label{theta}
\end{equation}
and
\begin{equation}
\boldsymbol{e}=\frac{\boldsymbol{\sigma }_{\mathrm{e}}\left(0\right)}{%
\left\Vert \boldsymbol{\sigma }_{\mathrm{e}}\left(0\right)\right\Vert }.
\label{e}
\end{equation}

The initial value $\theta\left(0\right)$ of $\theta\left(t\right)$ can be
obtained by~(\ref{theta}). By designing an attitude maneuver controller, the
rigid spacecraft is driven to rotate about the fixed Euler axis $\boldsymbol{%
e}$ in~(\ref{e}), such that the rotation angle $\theta\left(t\right)$
converges from the initial value $\theta\left(0\right)$ to approach $0$ or $%
2\pi$.

\subsection{Unwinding phenomenon}

The phenomenon that spacecraft performs a rotation angle larger than $\pi$
to arrive at the desired attitude is called the "unwinding". For the
rest-to-rest attitude maneuver error dynamics of a rigid spacecraft
described by~(\ref{systemmodel}) with~(\ref{theta}), $\theta\left(t\right)=0$
and $\theta \left(t\right)=2\pi$ represent the same attitude. However, the
existing attitude maneuver control schemes are designed to ensure that the
rotation angle $\theta\left(t\right)$ converges from any initial value $%
\theta\left(0\right)$ to $0$. In this case, if $\theta\left(0\right)>\pi$,
then the spacecraft is driven to perform a rotation larger than $\pi $ about
the Euler axis $\boldsymbol{e}$, which results in the unwinding phenomenon.
However, the rigid spacecraft can reach the desired attitude by rotating an
angle $2\pi-\theta\left(0\right)$, which is smaller than $\pi$, about the
Euler axis $\boldsymbol{e}$ in the opposite direction.

\subsection{Control objective}

\label{controlgoal}
The control task in this work is to design an unwinding-free attitude
sliding mode controller for the attitude maneuver error dynamics~(\ref%
{systemmodel}) with~(\ref{theta}) of a rigid spacecraft, such that the
following relations hold,
\begin{align}
\lim_{t\rightarrow \infty }\theta\left(t\right)& =0\ \mathrm{or }%
\lim_{t\rightarrow \infty }\theta\left(t\right)=2\pi,\ \lim_{t\rightarrow
\infty }\boldsymbol{\omega }_{\mathrm{e}}=\boldsymbol{0}.  \label{aim}
\end{align}
Moreover, the unwinding phenomenon is also avoided.%

\section{Controller design methods}

In this section, we aim to develop an unwinding-free attitude controller for
the system~(\ref{systemmodel}), using sliding mode control theory. To
facilitate the controller development, we give some lemmas in subsection~\ref%
{lemmas}. To avoid the unwinding phenomenon when the system states are on
the sliding surface, we construct a novel switching function in subsection~%
\ref{switching function}. In order to avoid the unwinding phenomenon before
the system states reach the sliding surface, a sliding mode control law with
a dynamic parameter is developed in subsection~\ref{controller}.%

\subsection{Some lemmas}

\label{lemmas}

\begin{lemma}\label{thetadynamic}
Consider the rotation angle $\theta\left(t\right)$ given by~(\ref{theta}). The following relation holds,
\begin{equation}
\dot{\theta}\left(t\right)=\boldsymbol{e}^{\mathrm{T}}\boldsymbol{\omega }_{\mathrm{e}},
\label{dottheta}
\end{equation}
where the attitude angular velocity error $\boldsymbol{\omega }_{\mathrm{e}}$ is defined in~(\ref{omega_e}), and the Euler axis $\boldsymbol{e}$ can be obtained from~(\ref{e}).
\end{lemma}

\begin{proof}
Taking the derivative of both sides of~(\ref{theta}), and using the first
relation of~(\ref{systemmodel}), yields
\begin{align}
\dot{\theta}\left( t\right) & =\frac{4\boldsymbol{e}^{\mathrm{T}}\dot{%
\boldsymbol{\sigma }}_{\mathrm{e}}}{1+\left( \boldsymbol{e}^{\mathrm{T}}%
\boldsymbol{\sigma }_{\mathrm{e}}\right) ^{2}}  \notag \\
& =\frac{4\boldsymbol{e}^{\mathrm{T}}M\left( \boldsymbol{\sigma }_{\mathrm{e}%
}\right) \boldsymbol{\omega }_{\mathrm{e}}}{1+\left( \boldsymbol{e}^{\mathrm{%
T}}\boldsymbol{\sigma }_{\mathrm{e}}\right) ^{2}}.  \label{thetainter}
\end{align}%
In addition, by replacing $\boldsymbol{x}$ with $\boldsymbol{\sigma }_{%
\mathrm{e}}$ in~(\ref{M}), and using~(\ref{sigma}), we have
\begin{align}
&\boldsymbol{e}^{\mathrm{T}}M\left( \boldsymbol{\sigma }_{\mathrm{e}}\right)
\notag \\
&=\!\frac{\left( 1-\left\Vert \boldsymbol{\sigma }_{\mathrm{e}}\right\Vert
^{2}\right) \boldsymbol{e}^{\mathrm{T}}+2\boldsymbol{e}^{\mathrm{T}}%
\boldsymbol{e}^{\times }\tan \frac{\theta \left( t\right) }{4}+2\boldsymbol{e%
}^{\mathrm{T}}\tan ^{2}\frac{\theta \left( t\right) }{4}}{4}  \notag \\
& =\!\frac{\left( 1-\tan ^{2}\frac{\theta \left( t\right) }{4}\right)
\boldsymbol{e}^{\mathrm{T}}+2\boldsymbol{e}^{\mathrm{T}}\tan ^{2}\frac{%
\theta \left( t\right) }{4}}{4}  \notag \\
&=\frac{\left( 1+\tan ^{2}\frac{\theta \left( t\right) }{4}\right)
\boldsymbol{e}^{\mathrm{T}}}{4}.  \label{M2}
\end{align}%
It follows from~(\ref{sigma}) and~(\ref{thetainter}) that
\begin{align*}
\dot{\theta}\left( t\right) & =\frac{4}{1+\boldsymbol{\sigma }_{\mathrm{e}}^{%
\mathrm{T}}\boldsymbol{\sigma }_{\mathrm{e}}}\frac{1+\left\Vert \boldsymbol{%
\sigma }_{\mathrm{e}}\right\Vert ^{2}}{4}\boldsymbol{e}^{\mathrm{T}}%
\boldsymbol{\omega }_{\mathrm{e}} \\
& =\boldsymbol{e}^{\mathrm{T}}\boldsymbol{\omega }_{\mathrm{e}}.
\end{align*}%
Thus, the proof is completed.
\end{proof}

\begin{lemma}
\label{finitetime} Suppose $V(x)$ is a $C^{1}$ smooth positive-definite
function (defined on $U\subset \mathbb{R}^{n}$) and $\dot{V}(x)+\lambda V^{\alpha
}(x) $ is a negative semi-definite function on $U\subset \mathbb{R}^{n}$ for $\alpha
\in (0,1)$ and $\lambda \in \mathbb{R}^{+}$, then there exists an area $U_{0}\subset
\mathbb{R}^{n} $ such that any $V(x)$ which starts from $U_{0}\subset \mathbb{R}^{n}$ can
reach $V(x)\equiv 0$ in finite time. Moreover, if $T_{\mathrm{s}}$ is the
time needed to reach $V(x)\equiv 0$, then
\begin{equation*}
T_{\mathrm{s}}\leq \frac{V^{1-\alpha }(x_{0})}{\lambda \left( 1-\alpha
\right) },
\end{equation*}%
where $V(x_{0})$ is the initial value of $V(x)$.
\end{lemma}

\subsection{Switching function}

\label{switching function}

For the rest-to-rest attitude maneuver error dynamics~(\ref{systemmodel})
with~(\ref{theta}) of a rigid spacecraft, the switching function is designed
as
\begin{equation}
\boldsymbol{s}=\boldsymbol{\omega }_{\mathrm{e}}-\alpha \rho \left(
\boldsymbol{\sigma }_{\mathrm{e}}\right) \boldsymbol{\sigma }_{\mathrm{e}},
\label{s}
\end{equation}%
where $\alpha $ is a positive number, and%
\begin{equation}
\rho \left( \boldsymbol{\sigma }_{\mathrm{e}}\right) =\frac{\sinh g\left(
\boldsymbol{\sigma }_{\mathrm{e}}\right) }{1+\boldsymbol{\sigma }_{\mathrm{e}%
}^{\mathrm{T}}\boldsymbol{\sigma }_{\mathrm{e}}},  \label{rou}
\end{equation}%
with%
\begin{equation}
g\left( \boldsymbol{\sigma }_{\mathrm{e}}\right) =\arctan \boldsymbol{e}^{%
\mathrm{T}}\boldsymbol{\sigma }_{\mathrm{e}}-\frac{\pi }{4}.  \label{g}
\end{equation}

Next, a theorem is given to demonstrate that the control goal in~(\ref{aim})
can be achieved when $\boldsymbol{\omega }_{\mathrm{e}}$ and $\boldsymbol{%
\sigma }_{\mathrm{e}}$ are restricted to the sliding surface $\boldsymbol{s}=%
\boldsymbol{0}.$ Moreover, it is proven that the unwinding phenomenon is
conquered on the sliding surface.

Before given the theorem, we should give some properties of the functions $%
\cosh x$ and $\sinh x$ for $x\in \left[-\pi,\pi\right]$. The minimum value
of the function $\cosh x $ can be obtained when $x=0$, and the maximum value
of the function $\cosh x $ can be obtained when $x=-\pi$ or $x=\pi$. In
addition, for $x < 0 $, $\sinh x< 0$ holds, and for $x \geq0 $, $\sinh x\geq
0 $ holds.

\begin{theorem}
\label{systemstatesconvergence} Consider the rest-to-rest attitude maneuver error
dynamics~(\ref{systemmodel}) with~(\ref{theta}) for a rigid spacecraft. When the attitude errors $\boldsymbol{\sigma}_{%
\mathrm{e}}$ and $\boldsymbol{\omega}_{\mathrm{e}}$ are restricted to the sliding surface $%
\boldsymbol{s}=\boldsymbol{0}$, the following conclusions are obtained.

(i) The unwinding phenomenon is avoided.

(ii) The control goal in~(\ref{aim}) is attained.
\end{theorem}

\begin{proof}
Suppose that when $t=t_{s0}$ the system states reach the sliding surface $%
\boldsymbol{s}=\boldsymbol{0}$. Then, it can be obtained from~(\ref{s}) that%
\begin{equation}
\boldsymbol{\omega }_{\mathrm{e}}=\alpha \rho \left( \boldsymbol{\sigma }_{%
\mathrm{e}}\right) \boldsymbol{\sigma }_{\mathrm{e}}.
\label{sliding surface}
\end{equation}

To prove the conclusion (i), we need to prove that the following relations
hold,
\begin{equation}
\lim_{t\rightarrow \infty }\theta \left( t\right) =\left\{
\begin{array}{l}
0,\ \mathrm{if}\ \theta \left( t_{s0}\right) \in \left( 0,\pi \right) , \\
2\pi ,\ \mathrm{if}\ \theta \left( t_{s0}\right) \in \left( \pi ,2\pi
\right) .%
\end{array}%
\right.   \label{unwinding}
\end{equation}%
It can be obtained from (\ref{dottheta}) and (\ref{sliding surface}) that
\begin{equation}
\dot{\theta}\left( t\right) =\alpha \boldsymbol{e}^{\mathrm{T}}\rho \left(
\boldsymbol{\sigma }_{\mathrm{e}}\right) \boldsymbol{\sigma }_{\mathrm{e}}.
\end{equation}%

In the following, we rewrite $\dot{\theta}\left( t\right) $ in terms of the
rotation angle $\theta \left( t\right) $ and the Euler axis $\boldsymbol{e}$%
. First, using~(\ref{sigma}) and~(\ref{g}), $g\left(
\boldsymbol{\sigma }_{\mathrm{e}}\right) $ is rewritten as
\begin{align}
g\left( \theta \left( t\right) \right) & =\arctan \boldsymbol{e}^{\mathrm{T}}%
\boldsymbol{e}\tan \frac{\theta \left( t\right) }{4}-\frac{\pi }{4}  \notag
\\
& =\frac{\theta \left( t\right) }{4}-\frac{\pi }{4}.  \label{gtheta}
\end{align}%
Applying~(\ref{sigma}) and~(\ref{gtheta}) to~(\ref{rou}), gives
\begin{align}
\rho \left( \theta \left( t\right) \right) & =\frac{g\left( \theta \left(
t\right) \right) }{1+\boldsymbol{e}^{\mathrm{T}}\boldsymbol{e}\tan ^{2}\frac{%
\theta \left( t\right) }{4}}  \notag \\
& =\frac{\sinh \left( \frac{\theta \left( t\right) }{4}-\frac{\pi }{4}%
\right) }{1+\tan ^{2}\frac{\theta \left( t\right) }{4}}  \notag \\
& =\sinh \left( \frac{\theta \left( t\right) }{4}-\frac{\pi }{4}\right) \cos
^{2}\frac{\theta \left( t\right) }{4}.  \label{routheta}
\end{align}%
It follows from~(\ref{sigma}) and~(\ref{sliding surface}) that $\boldsymbol{%
\omega }_{\mathrm{e}}$ can be rewritten as
\begin{equation}
\boldsymbol{\omega }_{\mathrm{e}}=\alpha \boldsymbol{e}\sinh \left( \frac{%
\theta \left( t\right) }{4}-\frac{\pi }{4}\right) \cos ^{2}\frac{\theta
\left( t\right) }{4}\tan \frac{\theta \left( t\right) }{4}.
\label{omegatheta}
\end{equation}
Following~(\ref{sigma}),~(\ref{dottheta}) and~(\ref{omegatheta}), one can obtain
\begin{align}
\dot{\theta}\left( t\right)
&=\alpha \sinh \left( \frac{\theta \left( t\right) }{4}-\frac{\pi }{4}%
\right) \cos ^{2}\frac{\theta \left( t\right) }{4}\tan \frac{\theta \left(
t\right) }{4}.\label{dottheta1}
\end{align}%

In addition, there hold $\sinh \left( \frac{\theta \left( t\right) }{4}-%
\frac{\pi }{4}\right) <0$ for $\theta \left( t\right) \in \left( 0,\pi
\right) $, and $\sinh \left( \frac{\theta \left( t\right) }{4}-\frac{\pi }{4}%
\right) >0$ for $\theta \left( t\right) \in \left( \pi ,2\pi \right) $. As $%
\tan \frac{\theta \left( t\right) }{4}\geq 0$ for $\theta \left( t\right) \in \left( 0 ,2\pi \right) $, it can be derived from~(\ref%
{dottheta1}) that $\dot{\theta}\left( t\right) <0$ for $\theta \left(
t_{s0}\right) \in \left( 0,\pi \right) $, $\dot{\theta}\left( t\right) >0$
for $\theta \left( t_{s0}\right) \in \left( \pi ,2\pi \right) $, and $\dot{%
\theta}\left( t\right) =0$ for $\theta\left(t\right) =0$ or $\theta\left(t\right) =2\pi $. Thus, the
relations in~(\ref{unwinding}) is obtained. This implies that the unwinding
phenomenon is conquered when the system states are sliding on the sliding
surface $\boldsymbol{s}=\boldsymbol{0}$.

Thus, (i) is proven.

Next, the fact that the control goal in~(\ref{aim}) is achieved on the sliding
surface $\boldsymbol{s}=\boldsymbol{0}$ is proven. For this end, we chose the following Lyapunov function,
\begin{equation}
V_{1}\left( t\right) =\kappa -\cosh g\left( \boldsymbol{\sigma }_{\mathrm{e}%
}\right) ,  \label{V1}
\end{equation}%
where $\kappa =\max \left( \cosh g\left( \boldsymbol{\sigma }_{\mathrm{e}%
}\right) \right) $. Substituting~(\ref{gtheta}) into~(\ref{V1}), yields
\begin{equation}
V_{1}\left( t\right) =\kappa -\cosh g\left( \theta \left( t\right) \right) ,
\label{V11}
\end{equation}%
where
\begin{equation*}
\kappa \!=\!\cosh \left( g\left( \theta \left( t\right) \right) \right)
|_{\theta =0}\!=\!\cosh \left( g\left( \theta \left( t\right) \right)
\right) |_{\theta =2\pi }.
\end{equation*}

The time derivative of $V_{1}\left( t\right) $ in~(\ref{V11}) along time is
\begin{equation*}
\dot{V}_{1}\left( t\right) =-\frac{\mathrm{d}g\left( \theta \left( t\right)
\right) }{\mathrm{d}t}\sinh g\left( \theta \left( t\right) \right) .
\end{equation*}%
It follows from~(\ref{gtheta}) and Lemma~\ref{thetadynamic} that
\begin{align}
\dot{V}_{1}\left( t\right) & =-\frac{\dot{\theta}\left( t\right) }{4}\sinh
\left( \frac{\theta \left( t\right) }{4}-\frac{\pi }{4}\right) .
\label{dotv11}
\end{align}%
By substituting~(\ref{omegatheta}) into~(\ref{dotv11}), we have
\begin{equation}
\dot{V}_{1}\left( t\right) =-\frac{\alpha }{4}\cos ^{2}\frac{\theta \left(
t\right) }{4}\sinh ^{2}\left( \frac{\theta \left( t\right) }{4}-\frac{\pi }{4%
}\right) \tan \frac{\theta \left( t\right) }{4}.  \label{dotV1}
\end{equation}%
It is clear that $\dot{V}_{1}\left( t\right) \leq 0$ because $\tan \frac{%
\theta \left( t\right) }{4}\geq 0$ holds for $\theta \in \left( 0,\ 2\pi
\right) $.

Moreover, according to~(\ref{dotV1}), we obtain that if $\dot{V}_{1}\left(
t\right) =0$, there hold
\begin{equation*}
\cos^2\frac{\theta\left(t\right)}{4}=0 ,
\end{equation*}
or
\begin{equation*}
\sinh^2 \left( \frac{\theta \left(t\right)}{4}-\frac{\pi }{4}\right)=0,
\end{equation*}
or
\begin{equation*}
\tan\frac{\theta \left( t\right) }{4}=0.
\end{equation*}
Thus, we have $\theta\left( t\right)=2\pi$ or $\theta\left( t\right)=\pi$ or
$\theta\left( t\right)=0$ if $\dot{V}_{1}\left(t\right)=0$. In addition, it
can be readily computed from~(\ref{V11}) that the maximum value of $%
V_{1}\left(t\right)$ is obtained at $\theta=\pi$, and the minimum value of $%
V_{1}\left(t\right)$ is obtained at $\theta\left(t\right)=2\pi$ or $%
\theta\left(t\right)=0$. Consequently, there holds%
\begin{equation*}
\lim_{t\rightarrow \infty }\theta \left(t\right)=0\ \mathrm{or}\
\lim_{t\rightarrow \infty }\theta \left(t\right)=2\pi.
\end{equation*}
Further, in view of~(\ref{omegatheta}), we have $\boldsymbol{\omega}_{%
\mathrm{e}}=\boldsymbol{0}$ for $\theta\left(t\right)=2\pi$ or $%
\theta\left(t\right)=0$. This implies that the control goal~(\ref{aim}) is
achieved on the sliding surface $\boldsymbol{s}=\boldsymbol{0}$.

Hence, the proof is completed.
\end{proof}
 By Theorem~\ref{systemstatesconvergence}, it is proven that the unwinding phenomenon is avoided when the system states are on the sliding surface $\boldsymbol{s}=\boldsymbol{0}$. In the subsequent subsection, it is shown that the unwinding-free performance of the closed-loop attitude maneuver system~(\ref{systemmodel}) is guaranteed by designing a sliding mode control law with a dynamic parameter.
\subsection{Unwinding-Free Sliding Mode Control Law}

\label{controller}

In this section, we need to construct a control law such that the condition $%
\boldsymbol{s}=\boldsymbol{0}$ is achieved in finite-time. This condition
assures us that all the system states of the closed-loop attitude maneuver
error dynamics~(\ref{systemmodel}) arrive at the sliding surface $%
\boldsymbol{s}=\boldsymbol{0}$ in finite-time. Moreover, the unwinding
phenomenon is also avoided before the system states reach the sliding
surface.

Consider a class of state feedback control for the rest-to-rest attitude
maneuver error dynamics~(\ref{systemmodel}) with~(\ref{theta}) of a rigid
spacecraft in the following form,
\begin{equation}
\boldsymbol{u}=\boldsymbol{u}_{\mathrm{eq}}+\boldsymbol{u}_{\mathrm{n}},
\label{u}
\end{equation}%
where $\boldsymbol{u}_{\mathrm{eq}}$ is the equivalent control term for the
nominal system, $\boldsymbol{u}_{\mathrm{n}}$ is the switching control term,
which is designed to deal with the disturbance. Thus, the equivalent control
$\boldsymbol{u}_{\mathrm{eq}}$ can be obtained from the nominal system part
by setting $\dot{\boldsymbol{s}}=\boldsymbol{0}$, such that%
\begin{equation}
\dot{\boldsymbol{s}}=\dot{\boldsymbol{\omega }}_{\mathrm{e}}-\alpha \dot{\rho%
}\left( \boldsymbol{\sigma }_{\mathrm{e}}\right) \boldsymbol{\sigma }_{%
\mathrm{e}}-\alpha \rho \left( \boldsymbol{\sigma }_{\mathrm{e}}\right) \dot{%
\boldsymbol{\sigma }}_{\mathrm{e}}=\boldsymbol{0}.  \label{dots}
\end{equation}%
By setting $\boldsymbol{d}=\boldsymbol{0}$, the following nominal system
part of~(\ref{systemmodel}) can be obtained,%
\begin{equation*}
J\dot{\boldsymbol{\omega }}_{\mathrm{e}}=-\boldsymbol{\omega }_{\mathrm{e}%
}^{\times }J\boldsymbol{\omega }_{\mathrm{e}}+\boldsymbol{u}.
\end{equation*}%
Substituting the above equation into~(\ref{dots}) obtains
\begin{equation*}
\dot{\boldsymbol{s}}=J^{-1}\left( -\boldsymbol{\omega }_{\mathrm{e}}^{\times
}J\boldsymbol{\omega }_{\mathrm{e}}+\boldsymbol{u}_{\mathrm{eq}}\right)
-\alpha \dot{\rho}\left( \boldsymbol{\sigma }_{\mathrm{e}}\right)
\boldsymbol{\sigma }_{\mathrm{e}}-\alpha \rho \left( \boldsymbol{\sigma }_{%
\mathrm{e}}\right) \dot{\boldsymbol{\sigma }}_{\mathrm{e}}=0.
\end{equation*}%
By solving the above equation concerning $\boldsymbol{u}_{\mathrm{eq}}$, we
have
\begin{equation*}
\boldsymbol{u}_{\mathrm{eq}}=\boldsymbol{\omega }_{\mathrm{e}}^{\times }J%
\boldsymbol{\omega }_{\mathrm{e}}+\alpha J\dot{\rho}\left( \boldsymbol{%
\sigma }_{\mathrm{e}}\right) \boldsymbol{\sigma }_{\mathrm{e}}+\alpha J\rho
\left( \boldsymbol{\sigma }_{\mathrm{e}}\right) \dot{\boldsymbol{\sigma }}_{%
\mathrm{e}}.
\end{equation*}%
In addition, the control term $\boldsymbol{u}_{\mathrm{n}}$ is designed as,
\begin{equation}
\boldsymbol{u}_{\mathrm{n}}=-\left( \gamma _{1}+\gamma _{2}\left( t\right)
\right) \mathrm{sgn}\left( \boldsymbol{s}\right) ,  \label{u_n}
\end{equation}%
where $\gamma _{1}\geq \left\Vert \boldsymbol{d}\right\Vert _{\max }$, $%
\gamma _{2}\left( t\right) $ is a positive-valued function, and
\begin{equation*}
\mathrm{sgn}\left( \boldsymbol{s}\right) =\left[ \frac{s_{1}}{\left\vert
s_{1}\right\vert }\ \frac{s_{2}}{\left\vert s_{2}\right\vert }\ \frac{s_{3}}{%
\left\vert s_{3}\right\vert }\right] ^{\mathrm{T}}.
\end{equation*}

By concluding previous derivations, the following unwinding-free sliding
mode control (briefly, UFSMC) law is obtained,%
\begin{equation}
\left\{
\begin{array}{l}
\boldsymbol{u}=\boldsymbol{u}_{\mathrm{eq}}+\boldsymbol{u}_{\mathrm{n}}, \\
\boldsymbol{u}_{\mathrm{eq}}=\boldsymbol{\omega }_{\mathrm{e}}^{\times }J%
\boldsymbol{\omega }_{\mathrm{e}}+\alpha J\dot{\rho}\left( \boldsymbol{%
\sigma }_{\mathrm{e}}\right) \boldsymbol{\sigma }_{\mathrm{e}}+\alpha J\rho
\left( \boldsymbol{\sigma }_{\mathrm{e}}\right) \dot{\boldsymbol{\sigma }}_{%
\mathrm{e}}, \\
\boldsymbol{u}_{\mathrm{n}}=-\left( \gamma _{1}+\gamma _{2}\left( t\right)
\right) \mathrm{sgn}\left( \boldsymbol{s}\right) , \\
\boldsymbol{s}=\boldsymbol{\omega }_{\mathrm{e}}-\alpha \rho \left(
\boldsymbol{\sigma }_{\mathrm{e}}\right) \boldsymbol{\sigma }_{\mathrm{e}},
\\
\rho \left( \boldsymbol{\sigma }_{\mathrm{e}}\right) =\frac{\sinh g\left(
\boldsymbol{\sigma }_{\mathrm{e}}\right) }{1+\boldsymbol{\sigma }_{\mathrm{e}%
}^{\mathrm{T}}\boldsymbol{\sigma }_{\mathrm{e}}}, \\
g\left( \boldsymbol{\sigma }_{\mathrm{e}}\right) =\arctan \boldsymbol{e}^{%
\mathrm{T}}\boldsymbol{\sigma }_{\mathrm{e}}-\frac{\pi }{4}, \\
\boldsymbol{e}=\frac{\boldsymbol{\sigma }_{\mathrm{e}}\left( 0\right) }{%
\left\Vert \boldsymbol{\sigma }_{\mathrm{e}}\left( 0\right) \right\Vert },%
\end{array}%
\right.  \label{unwindingSMC}
\end{equation}%
where $\alpha $ is a positive number, $\boldsymbol{\sigma }_{\mathrm{e}%
}\left( 0\right) $ is the initial value of $\boldsymbol{\sigma }_{\mathrm{e}%
},$ $\gamma _{1}\geq \left\Vert \boldsymbol{d}\right\Vert _{\max }$, and $%
\gamma _{2}\left( t\right) $ is a positive-valued function, which is given
in the following theorem.

\begin{theorem}
\label{switchingconvergence} Consider a rest-to-rest attitude maneuver problem of a rigid spacecraft described by~(\ref%
{systemmodel}) with~(\ref{theta}). If the dynamic parameter $\gamma
_{2}\left( t\right) $ is chosen as
\begin{equation}
\gamma _{2}\left( t\right) =\frac{\alpha  }{\lambda _{\min }\left(
J^{-1}\right) }\left\vert \dot{h}\left( t\right) \right\vert ,
\label{gamma2}
\end{equation}%
where $\alpha>0 $, $\lambda _{\min }\left(
J^{-1}\right)$ represents the minimum eigenvalue of the inverse matrix of $J$, and%
\begin{equation}
h\left( t\right) =\rho \left( \boldsymbol{\sigma }_{\mathrm{e}}\right)\left\Vert\boldsymbol{\sigma }_{\mathrm{e}}\right\Vert,  \label{h}
\end{equation}%
with $\rho \left( \boldsymbol{\sigma }_{\mathrm{e}}\right)$ defining in~(\ref{rou}).
Then, the following conclusions are acquired.

(i) The switching function $\boldsymbol{s}$ converges to zero in finite time.

(ii) The unwinding phenomenon is avoided before the system states reach
the switching surface $\boldsymbol{s}=\boldsymbol{0}$.
\end{theorem}

\begin{proof}
To prove (i), we chose the following Lyapunov function,
\begin{equation}
V_{2}\left( t\right) =\frac{1}{2}\boldsymbol{s}^{\mathrm{T}}\boldsymbol{s}.
\label{v2}
\end{equation}%
Taking time derivative of the above equation, and using~(\ref{s}), yields
\begin{align*}
\dot{V}_{2}\left( t\right) & =\boldsymbol{s}^{\mathrm{T}}\dot{\boldsymbol{s}}
\\
& =\boldsymbol{s}^{\mathrm{T}}\left( \dot{\boldsymbol{\omega }}_{\mathrm{e}%
}-\alpha \dot{\rho}\left( \boldsymbol{\sigma }_{\mathrm{e}}\right)
\boldsymbol{\sigma }_{\mathrm{e}}-\alpha \rho \left( \boldsymbol{\sigma }_{%
\mathrm{e}}\right) \boldsymbol{\dot{\sigma}}_{\mathrm{e}}\right) .
\end{align*}%
Substituting the second equation of (\ref{systemmodel}) and controller~(\ref%
{unwindingSMC}) into the above equation, we arrive at%
\begin{align}
\dot{V}_{2}\left( t\right) & =\boldsymbol{s}^{\mathrm{T}}\left( J^{-1}\left(
-\boldsymbol{\omega }_{\mathrm{e}}^{\times }J\boldsymbol{\omega }_{\mathrm{e}%
}+\boldsymbol{u}+\boldsymbol{d}\right) -\alpha \dot{\rho}\left( \boldsymbol{%
\sigma }_{\mathrm{e}}\right) \boldsymbol{\sigma }_{\mathrm{e}}\right.  \notag
\\
& \quad \left. -\alpha \rho \left( \boldsymbol{\sigma }_{\mathrm{e}}\right)
\dot{\boldsymbol{\sigma }}_{\mathrm{e}}\right)  \notag \\
& =-\boldsymbol{s}^{\mathrm{T}}J^{-1}\left( \gamma _{1}+\gamma _{2}\left(
t\right) \right) \mathrm{sgn}\left( \boldsymbol{s}\right) +\boldsymbol{s}^{%
\mathrm{T}}J^{-1}\boldsymbol{d}  \notag \\
& =-\gamma _{2}\left( t\right) \boldsymbol{s}^{\mathrm{T}}J^{-1}\mathrm{sgn}%
\left( \boldsymbol{s}\right) -\boldsymbol{s}^{\mathrm{T}}J^{-1}\left( \gamma
_{1}-\left\Vert \boldsymbol{d}\right\Vert \right)  \notag \\
& \leq -\gamma _{2}\left( t\right) \boldsymbol{s}^{\mathrm{T}}J^{-1}\mathrm{%
sgn}\left( \boldsymbol{s}\right) .  \label{dotv21}
\end{align}%
Obviously, there holds%
\begin{equation}
\boldsymbol{s}^{\mathrm{T}}J^{-1}\mathrm{sgn}\left( \boldsymbol{s}\right)
\geq \lambda _{\min }\left( J^{-1}\right) \left\Vert \boldsymbol{s}%
\right\Vert .  \label{inequ}
\end{equation}%
According to~(\ref{inequ}), one deduces from~(\ref{dotv21}) that
\begin{equation*}
\dot{V}_{2}\left( t\right) \leq -\gamma _{2}\left( t\right) \lambda _{\min
}\left( J^{-1}\right) \left\Vert \boldsymbol{s}\right\Vert .
\end{equation*}%
By~(\ref{v2}) and~(\ref{inequ}), the above equation can be rewritten as%
\begin{align}
\dot{V}_{2}\left( t\right) & \leq -\sqrt{2}\gamma _{2}\left( t\right)
\lambda _{\min }\left( J^{-1}\right) \left( \frac{1}{2}\boldsymbol{s}^{%
\mathrm{T}}\boldsymbol{s}\right) ^{\frac{1}{2}}  \notag \\
& =-\sqrt{2}\gamma _{2}\left( t\right) \lambda _{\min }\left( J^{-1}\right)
V_{2}^{\frac{1}{2}}\left( t\right) .  \label{dotv2}
\end{align}%
Using Lemma~\ref{finitetime}, it is immediate to conclude that the switching
function $\boldsymbol{s}$ converges to zero in finite time.

Thus, (i) is proven.

Next, we prove that the unwinding-free performance is ensured by the
developed controller~(\ref{unwindingSMC}) with~(\ref{gamma2}) when the
system states are outside the switching surface $\boldsymbol{s}=\boldsymbol{0%
}$.

It can be further derived from~(\ref{dotv2}) that
\begin{equation*}
\frac{\dot{V}_{2}\left( t\right) }{V_{2}^{\frac{1}{2}}\left( t\right) }\leq -%
\sqrt{2}\gamma _{2}\left( t\right) \lambda _{\min }\left( J^{-1}\right) .
\end{equation*}%
Suppose the initial time is $t_{0}=0$. By taking integral of both sides of
the above equation, we have
\begin{equation*}
\int_{0}^{t}\frac{\dot{V}_{2}\left( \tau \right) }{V_{2}^{\frac{1}{2}}\left(
\tau \right) }\mathrm{d}\tau \leq -\sqrt{2}\lambda _{\min }\left(
J^{-1}\right) \int_{0}^{t}\gamma _{2}\left( \tau \right) \mathrm{d}\tau ,
\end{equation*}%
or, equivalently,
\begin{equation}
V_{2}^{\frac{1}{2}}\left( t\right) \leq -\frac{\lambda _{\min }\left(
J^{-1}\right) }{\sqrt{2}}\int_{0}^{t}\gamma _{2}\left( \tau \right) \mathrm{d%
}\tau +V_{2}^{\frac{1}{2}}\left( 0\right) .  \label{v2inequ}
\end{equation}

Let
\begin{equation}
v\left( t\right) =\boldsymbol{e}^{\mathrm{T}}\boldsymbol{s}.  \label{vinter}
\end{equation}%
Then, applying~(\ref{s}),~(\ref{sigma}), Lemma~\ref{thetadynamic}, and~(\ref%
{h}) to~(\ref{vinter}), yields
\begin{align}
v\left( t\right) & =\boldsymbol{e}^{\mathrm{T}}\boldsymbol{\omega }_{\mathrm{%
e}}-\alpha \rho \left( \boldsymbol{\sigma }_{\mathrm{e}}\right) \boldsymbol{e%
}^{\mathrm{T}}\boldsymbol{\sigma }_{\mathrm{e}}  \notag \\
& =\dot{\theta}\left( t\right) -\alpha \rho \left( \boldsymbol{\sigma }_{%
\mathrm{e}}\right) \left\Vert \boldsymbol{\sigma }_{\mathrm{e}}\right\Vert
\notag \\
& =\dot{\theta}\left( t\right) -\alpha h\left( t\right) .  \label{v}
\end{align}%
In addition, it can be obtained from~(\ref{vinter}) that
\begin{align*}
v^{2}\left( t\right) & =\left( \boldsymbol{e}^{\mathrm{T}}\boldsymbol{s}%
\right) ^{\mathrm{T}}\boldsymbol{e}^{\mathrm{T}}\boldsymbol{s} \\
& \leq \left\Vert \boldsymbol{ee}^{\mathrm{T}}\right\Vert \left\Vert
\boldsymbol{s}\right\Vert ^{2} \\
& \leq \lambda _{\max }\left( \boldsymbol{ee}^{\mathrm{T}}\right) \left\Vert
\boldsymbol{s}\right\Vert ^{2},
\end{align*}%
where $\lambda _{\max }\left( \boldsymbol{ee}^{\mathrm{T}}\right)$ is the maximum eigenvalue of the matrix $ \boldsymbol{ee}^{\mathrm{T}}$. Note that the Euler axis $\boldsymbol{e}$ is a unit vector, thus the matrix $%
\boldsymbol{ee}^{\mathrm{T}}$ is an idempotent matrix. Consequently, we have
$\lambda _{\max }\left( \boldsymbol{ee}^{\mathrm{T}}\right) =1$. Then, it is
clear that%
\begin{equation*}
v^{2}\left( t\right) \leq \left\Vert \boldsymbol{s}\right\Vert ^{2}.
\end{equation*}%
This together with~(\ref{v2}), results in%
\begin{equation}
\frac{1}{2}v^{2}\left( t\right) \leq V_{2}\left( t\right) .  \label{V2t}
\end{equation}

In this paper, the rest-to-rest attitude maneuver problem is considered,
thus the initial attitude velocity is zero, i.e., $\boldsymbol{\omega }_{%
\mathrm{e}}\left( 0\right) =\boldsymbol{0}$. Further, it can be obtained
from~(\ref{dottheta}) that $\dot{\theta}\left( t\right) =0$. In this case,
by (\ref{v}), the initial value of $v\left( t\right) $ can be obtained as%
\begin{equation}
v\left( 0\right) =-\alpha h\left( 0\right) .  \label{v0}
\end{equation}%
As $\boldsymbol{\omega }_{\mathrm{e}}\left( 0\right) =\boldsymbol{0}$, then
by~(\ref{v2}),~(\ref{s}),~(\ref{h}), and~(\ref{v0}), the initial value of $%
V_{2}\left( 0\right) $ can be obtained as
\begin{align}
V_{2}\left( 0\right) & =\frac{1}{2}\boldsymbol{s}^{\mathrm{T}}\left(
0\right) \boldsymbol{s}\left( 0\right)  \notag \\
& =\frac{1}{2}\alpha ^{2}\rho ^{2}\left( \boldsymbol{\sigma }_{\mathrm{e}%
}\left( 0\right) \right) \left\Vert \boldsymbol{\sigma }_{\mathrm{e}}\left(
0\right) \right\Vert ^{2}  \notag \\
& =\frac{1}{2}v^{2}\left( 0\right) .  \label{V20}
\end{align}%
Substituting~(\ref{V2t}) and~(\ref{V20}) into~(\ref{v2inequ}) gives
\begin{align}
\left( \frac{1}{2}v^{2}\left( t\right) \right) ^{\frac{1}{2}}& \leq V^{^{%
\frac{1}{2}}}\left( t\right) \leq -\frac{\lambda _{\min }\left(
J^{-1}\right) }{\sqrt{2}}\int_{0}^{t}\gamma _{2}\left( \tau \right) \mathrm{d%
}\tau  \notag \\
& \qquad\qquad\quad +\left( \frac{1}{2}v^{2}\left( 0\right) \right) ^{\frac{1%
}{2}},  \label{vinequ}
\end{align}%
which can be rewritten as
\begin{equation}
\left\vert v\left( t\right) \right\vert \!\leq \!-\lambda _{\min }\left(
J^{-1}\right) \int_{0}^{t}\gamma _{2}\left( \tau \right) \mathrm{d}\tau
\!+\!\left\vert v\left( 0\right) \right\vert .  \label{inequ1}
\end{equation}%
As $\gamma _{2}\left( t\right) >0$, it can be obtained from~(\ref{inequ1})
that $v\left( t\right) $ will decrease to $0$ when $v\left( 0\right) >0$,
and $v\left( t\right) $ will increase to $0$ when $v\left( 0\right) <0$.

To prove the unwinding-free phenomenon of the proposed control law~(\ref%
{unwindingSMC}) with $\gamma _{2}\left( t\right) $ being chosen in~(\ref%
{gamma2}), we need to prove that $\dot{\theta}\left( t\right) < 0$ for $%
\theta \left(0\right)\in \left( 0,\pi \right) $, and $\dot{\theta}\left(
t\right) > 0$ for $\theta \left(0\right)\in \left( \pi ,2\pi \right) $. To
this end, the following two cases are considered to complete the proof.

(1) When $\theta \left( 0\right) \!\in \!\left( 0,\pi \right) ,$ there holds
$\sinh \left( \frac{\theta \left( 0\right) }{4}-\frac{\pi }{4}\right) <0$.
Then, it can be obtained from~(\ref{sigma}), ~(\ref{routheta}), and~(\ref{h}%
) that there holds $h\left( 0\right) =\rho \left( \theta \left( 0\right)
\right) \tan \frac{\theta \left( 0\right) }{4}<0$. Further, according to~(%
\ref{v0}), $v\left( 0\right) =-\alpha h\left( 0\right) >0$ holds. Thus, $%
v\left( t\right) $ will decrease to zero. In such a case, by using~(\ref{v})
and~(\ref{v0}),~(\ref{inequ1}) can be rewritten as
\begin{equation*}
\dot{\theta}\left( t\right) -\alpha h\left( t\right) \leq -\lambda _{\min
}\left( J^{-1}\right) \int_{0}^{t}\gamma _{2}\left( \tau \right) \mathrm{d}%
\tau -\alpha h\left( 0\right) .
\end{equation*}%
It can be further rewritten as
\begin{align}
\dot{\theta}\left( t\right) & \leq -\lambda _{\min }\left( J^{-1}\right)
\int_{0}^{t}\gamma _{2}\left( \tau \right) \mathrm{d}\tau -\alpha h\left(
0\right) +\alpha h\left( t\right)  \notag \\
& =-\lambda _{\min }\left( J^{-1}\right) \int_{0}^{t}\gamma _{2}\left( \tau
\right) \mathrm{d}\tau +\alpha \int_{0}^{t}\frac{\mathrm{d}h\left( \tau
\right) }{\mathrm{d}\tau }d\tau  \notag \\
& =-\int_{0}^{t}\left( \lambda _{\min }\left( J^{-1}\right) \gamma
_{2}\left( \tau \right) -\alpha \dot{h}\left( \tau \right) \right) \mathrm{d}%
\tau .  \label{case1theta}
\end{align}

If $\dot{h}\left( t\right) >0$, then it can be obtained from~(\ref{gamma2})
that $\gamma _{2}\left( t\right) \!=\!\frac{\alpha \dot{h}\left( t\right) }{%
\lambda _{\min }\left( J^{-1}\right) }.$ It is followed from~(\ref%
{case1theta}) that $\dot{\theta}\left( t\right) \leq 0.$

If $\dot{h}\left( t\right) <0,$then it can be obtained from~(\ref{gamma2})
that $\gamma _{2}\left( t\right) =-\frac{\alpha \dot{h}\left( t\right) }{%
\lambda _{\min }\left( J^{-1}\right) }$. With this, it can be derived from (%
\ref{case1theta}) that%
\begin{equation*}
\dot{\theta}\left( t\right) \leq 2\alpha \int_{0}^{t}\dot{h}\left( \tau
\right) \mathrm{d}\tau \leq 0.
\end{equation*}

Thus, it can be obtained from above two cases that when $\theta\left(0%
\right)\in\left(0,\pi\right)$, there holds $\dot{\theta}\left(t\right)\leq 0$%
.

(2) When $\theta \left( 0\right) \!\in \!\left( \pi ,2\pi \right) ,$ there
holds $\sinh \left( \frac{\theta \left( 0\right) }{4}-\frac{\pi }{4}\right)
<0$. Then, it can be obtained from~(\ref{sigma}), ~(\ref{routheta}), and~(%
\ref{h}) that there holds $h\left( 0\right) =\rho \left( \theta \left(
0\right) \right) \tan \frac{\theta \left( 0\right) }{4}>0$. Further,
according to~(\ref{v0}), $v\left( 0\right) =-\alpha h\left( 0\right) <0$
holds. Thus, $v\left( t\right) $ will increase to zero. In such a case, by
using~(\ref{v}) and~(\ref{v0}),~(\ref{inequ1}) can be rewritten as
\begin{equation*}
-\dot{\theta}\left( t\right) +\alpha h\left( t\right) \leq -\lambda _{\min
}\left( J^{-1}\right) \int_{0}^{t}\gamma _{2}\left( \tau \right) \mathrm{d}%
\tau +\alpha h\left( 0\right) .
\end{equation*}%
Then, the following equation can be further obtained,%
\begin{align}
\dot{\theta}\left( t\right) & \geq \lambda _{\min }\left( J^{-1}\right)
\int_{0}^{t}\gamma _{2}\left( \tau \right) \mathrm{d}\tau +\alpha h\left(
t\right) -\alpha h\left( 0\right)  \notag \\
& =\lambda _{\min }\left( J^{-1}\right) \int_{0}^{t}\gamma _{2}\left( \tau
\right) \mathrm{d}\tau +\alpha \int_{0}^{t}\frac{\mathrm{d}h\left( \tau
\right) }{\mathrm{d}\tau }\mathrm{d}\tau  \notag \\
& \geq \int_{0}^{t}\left( \lambda _{\min }\left( J^{-1}\right) \gamma
_{2}\left( \tau \right) +\alpha \dot{h}\left( \tau \right) \right) \mathrm{d}%
\tau .  \label{case2theta}
\end{align}

If $\dot{h}\left( t\right) >0$, there holds $\gamma _{2}\left( \tau \right)
\!=\!\frac{\alpha \dot{h}\left( t\right) }{\lambda _{\min }\left(
J^{-1}\right) }$ according to~(\ref{gamma2}). Substituting it into~(\ref%
{case2theta}), we have%
\begin{equation*}
\dot{\theta}\left( t\right) \geq 2\alpha \int_{0}^{t}\dot{h}\left( \tau
\right) \mathrm{d}\tau \geq 0.
\end{equation*}

If $\dot{h}\left( t\right) <0,$ there holds $\gamma _{2}\left( \tau \right)
\!=\!-\frac{\alpha \dot{h}\left( t\right) }{\lambda _{\min }\left(
J^{-1}\right) }$ from~(\ref{gamma2}). Substituting it into~(\ref{case2theta}%
), we have $\dot{\theta}\left( t\right) \geq 0.$\

Thus, it can be obtained from above two cases that when $\theta\left(0%
\right)\in\left(\pi,2\pi\right)$, the rotation angle $\theta\left(t\right)$
will increase to $2\pi$.

Based on the above discussion, the conclusion (ii) is proven.%
\end{proof}

In Theorem~\ref{switchingconvergence}, the unwinding-free performance before
the system states reach the switching surface is proven. In Theorem~\ref%
{systemstatesconvergence}, the unwinding-free performance when the system
states are constricted to the sliding surface is also shown. The results in
these two theorems have illustrated that the proposed UFSMC law~(\ref%
{unwindingSMC}) has the unwinding-free property.

\begin{remark}
One drawback of the control law~(\ref{u_n}) is that it is discontinuous due to the discontinuousness of $\boldsymbol{u}_{\mathrm{n}}$ about
the sliding surface $\boldsymbol{s}=\boldsymbol{0}$. This characteristic may cause an
undesirable chattering phenomenon. For practical implementations, the
controller must be smoothed. Thus, the discontinuous function $\mathrm{sgn}%
\left( \boldsymbol{s}\right) $ is replaced by the smooth continuous function
$\boldsymbol{l}\left( \boldsymbol{s}\right)=\left[ l\left( s_{1}\right) \
l\left( s_{2}\right) \ l\left( s_{3}\right) \right] ^{\mathrm{T}}$ with $%
l\left( s_{i}\right) $ in the following equation,
\begin{equation}
l\left( s_{i}\right) =\left\{
\begin{array}{lc}
\mathrm{sgn}\left( s_{i}\right) , & \mbox{if }\left\vert s_{i}\right\vert
\geq \varepsilon_1 , \\
\mathrm{\arctan }\frac{s_{i}\tan \left( 1\right) }{\varepsilon_1 }, &
\mbox{if
}\left\vert s_{i}\right\vert <\varepsilon_1 ,%
\end{array}%
i=1,2,3,\right.  \label{f_si}
\end{equation}%
where $\varepsilon $ is a small positive value. As $\varepsilon_1 $ approaches
zero, the performance of this boundary layer can be made arbitrarily close
to that of original control law.

Another drawback of the proposed control law~(\ref{unwindingSMC}) is that it suffers the singular problem because $\theta\left(t\right)=2\pi$ is a singular point for $\boldsymbol{\sigma}_{\mathrm{e}}$, which may cause an unbounded control magnitude. This potential drawback can be resolved by introducing a boundary layer about $\boldsymbol{\sigma}_{\mathrm{e}}$, such that
\begin{equation*}
\sigma _{\mathrm{e}i}=\left\{
\begin{array}{l}
\frac{\mathrm{sgn}\left(\sigma_{\mathrm{e}i}\right)}{\varepsilon_2},\ \text{if}\ \frac{1}{%
\left\vert \sigma _{\mathrm{e}i}\right\vert }\leq \varepsilon_2, \\
\sigma _{\mathrm{e}i},\ \text{if}\ \frac{1}{\left\vert \sigma _{\mathrm{e}%
i}\right\vert }>\varepsilon_2,%
\end{array}%
\right.
\end{equation*}
where $\varepsilon_2>0$, $\sigma _{\mathrm{e}i}$ is the $i$-th element of the vector $\boldsymbol{\sigma}_{\mathrm{e}}$, and
\begin{equation*}
\mathrm{sgn}\left( \sigma _{\mathrm{e}i}\right) =\left\{
\begin{array}{l}
-1,\ \text{if}\ \sigma _{\mathrm{e}i}<0, \\
1,\ \text{if}\ \sigma _{\mathrm{e}i}>0.%
\end{array}%
\right.
\end{equation*}
As $\varepsilon_2$ approaches
zero, the rotation angle $\theta\left(t\right)$ can be driven arbitrarily close to $2\pi$ if the initial value $\theta\left(0\right)$ of $\theta\left(t\right)$ is larger than $\pi$.
\end{remark}

The advantage of the proposed UFSMC law~(\ref{unwindingSMC}) is that the
unwinding phenomenon is avoided during the rigid spacecraft attitude
maneuver, and the disturbance is compensated by the designed controller.

\section{Simulation}

In this section, simulations are conducted to demonstrate the performance of
the presented UFSMC law~(\ref{unwindingSMC}) for rest-to-rest attitude
maneuvers of a rigid spacecraft. In addition, the SMC controller in~\cite%
{crassidis1996sliding} is adopted for comparison.

\subsection{Simulation Settings}

\subsubsection{Spacecraft parameter values}

\label{system parameters} The inertia matrix of the rigid spacecraft is $J=%
\mathrm{diag}\left[ 114\ 86\ 87\right] \mathrm{kg\cdot m}^{2}$. The initial
values of the attitude velocity $\boldsymbol{\omega}$ and attitude $%
\boldsymbol{\sigma}$ are $\boldsymbol{\omega}\left(0\right)=\left[ 0\ 0\ 0%
\right] ^{\text{T}}$ and $\boldsymbol{\sigma}\left(0\right)=\left[ 0\ 0\ 0%
\right] ^{\text{T}}$, respectively. The disturbance is $\boldsymbol{d}%
=10^{-2}\times\left[ \sin \left( 0.05t\right)\ 0.5\sin \left( 0.05t\right)\
-\cos \left( 0.05t\right) \right] ^{\mathrm{T}}$.

\subsubsection{Controller parameter values}

\label{controller parameter values} The tuning parameters of the proposed
UFSMC law~(\ref{unwindingSMC}) are chosen as
\begin{equation*}
\alpha =2,\gamma _{1}=30,\ \varepsilon _{1}=0.5,\ \varepsilon _{2}=0.0001.
\end{equation*}%
In addition, $\gamma _{2}\left( t\right) $ is obtained from~(\ref{gamma2}).
The parameters of the SMC controller \cite{crassidis1996sliding} are chosen
as,
\begin{equation*}
k=1.5,\ \lambda =-0.5,\ \varepsilon =0.5.
\end{equation*}%
%
%
%
%
%
%
%
%
%
%
%
%

\subsubsection{Control goal}

\label{control goal} The control goal is to perform two rest-to-rest
attitude maneuvers for the rigid spacecraft with system parameters given in
Section~\ref{system parameters}.
Two different scenarios of desired attitude values are given as follows.

Scenario A. The desired attitude and angular velocity are $\boldsymbol{\sigma%
}_{\mathrm{d}}\!=\!\left[0.1\ 0.2\ -0.3 \right] ^{\mathrm{T}},$ and $%
\boldsymbol{\omega}_{\mathrm{d}} \!=\!\left[ 0\ 0\ 0\right] ^{\mathrm{T}}$%
rad/s, respectively.

Scenario B. The desired attitude and angular velocity are $\boldsymbol{\sigma%
}_{\mathrm{d}}\!=\!\left[0.7809\ 0.4685\ -0.7809 \right] ^{\mathrm{T}},$ and
$\boldsymbol{\omega}_{\mathrm{d}} =\left[ 0\ 0\ 0\right] ^{\mathrm{T}}$%
rad/s, respectively.

In Scenario A, it can be obtained from~(\ref{sigma_e}) that $\boldsymbol{%
\sigma}_{\mathrm{e}}\left(0\right)=\left[0.1\ 0.2\ -0.3\right] ^{\mathrm{T}}$%
. Further, there holds $\theta\left(0\right)=1.4321<\pi$ according to~(\ref%
{sigma}). Thus, $\theta\left(t\right)=0$ is the nearest equilibrium. The
controller needs to guarantee that $\theta\left(t\right)$ decreases to zero
monotonically. In Scenario B, it can be obtained from~(\ref{sigma_e}) that $%
\boldsymbol{\sigma}_{\mathrm{e}}\left(0\right)=\left[0.7809\ 0.4685\ -0.7809%
\right] ^{\mathrm{T}}$. Further, there holds $\theta\left(0\right)=3.5036>%
\pi $ according to~(\ref{sigma}). Thus, the spacecraft needs to rotate $%
3.5036$ rad to reach the desired attitude if only $\theta\left(t\right)=0$
is chosen as the equilibrium. However, the spacecraft only needs to rotate $%
2.7796$ rad if $\theta\left(t\right)=2\pi$ is also considered as an
equilibrium.

\subsection{Simulation results}

\subsubsection{Simulation results for Scenario A}

The SMC controller in~\cite{crassidis1996sliding} and proposed UFSMC law (%
\ref{unwindingSMC}) are adopted to do simulations for Scenario A. The
simulation results are shown in Fig. \ref{fig_2}.

The response of $\theta\left(t\right)$ and angular velocity error $\omega_{%
\mathrm{e}i}, i=1,2,3$ are shown in Fig.~\ref{fig_A1} and Fig.~\ref{fig_C1},
respectively. It can be seen from Fig. \ref{fig_A1} and Fig. \ref{fig_C1}
that $\theta\left(t\right)$ and the angular velocity errors of the system~(%
\ref{systemmodel}) converge to $0$ in about $6\mathrm{s}$ by adopting the
proposed UFSMC law (\ref{unwindingSMC}), while the SMC controller needs $12%
\mathrm{s}$. The spacecraft attitude responses using Euler angles, i.e.,
Roll, Pitch, and Yaw, are shown in Fig. \ref{fig_E1}, which indicates that
the attitude maneuver problem is effectively settled by the controller UFSMC
law (\ref{unwindingSMC}) and SMC law. 
The time evolution of control torques $u_{i},i=1,2,3$ are shown in Fig. \ref%
{fig_D1}. The control torque of the proposed UFSMC law (\ref{unwindingSMC})
is smaller than that of the SMC controller.

In conclusion, the UFSMC controller can obtain higher pointing accuracy and
better stability in a shorter time.
\begin{figure*}[!t]
\centering
\subfigure[Time response of rotation angle $\theta\left(t\right)$]{
		\label{fig_A1}
		\includegraphics[width=3in]{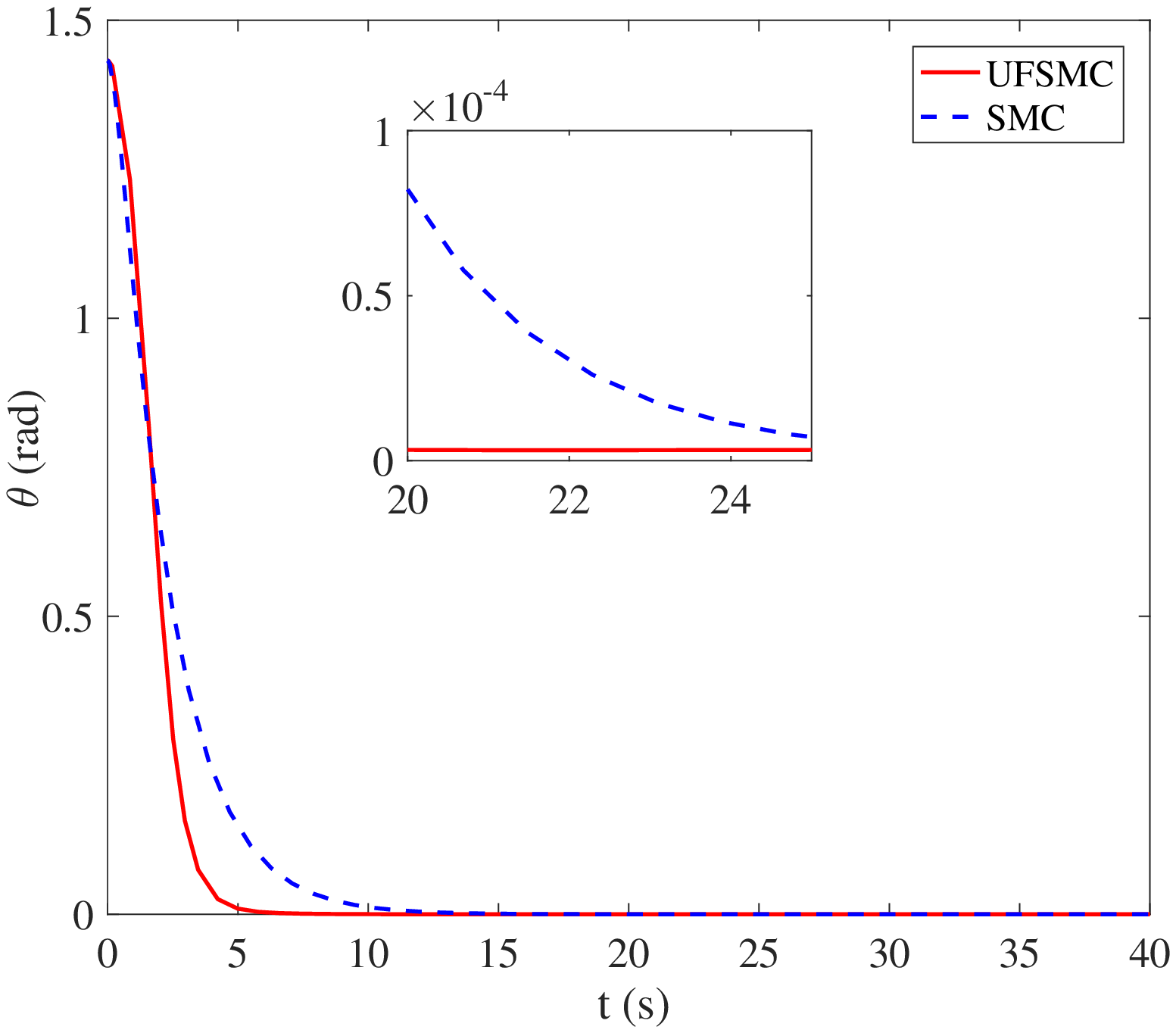}} 
\hspace{1cm}
\subfigure[Time response of the angular velocity]{
		\label{fig_C1}
		\includegraphics[width=3in]{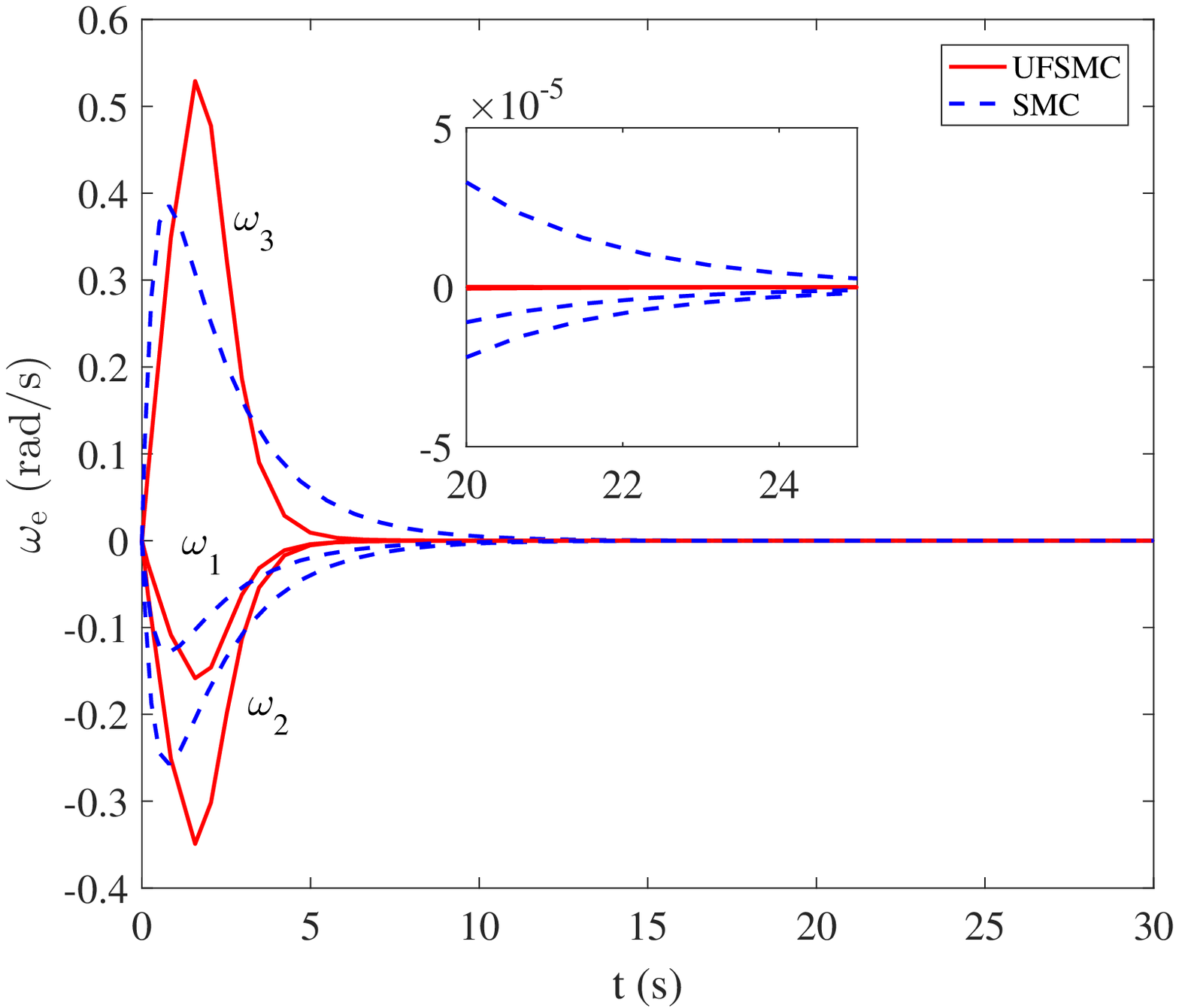}} 
\subfigure[Evolution of the Euler angles $\phi,\theta, \psi $ for the Scenario A]{
		\label{fig_E1}
		\includegraphics[width=3in]{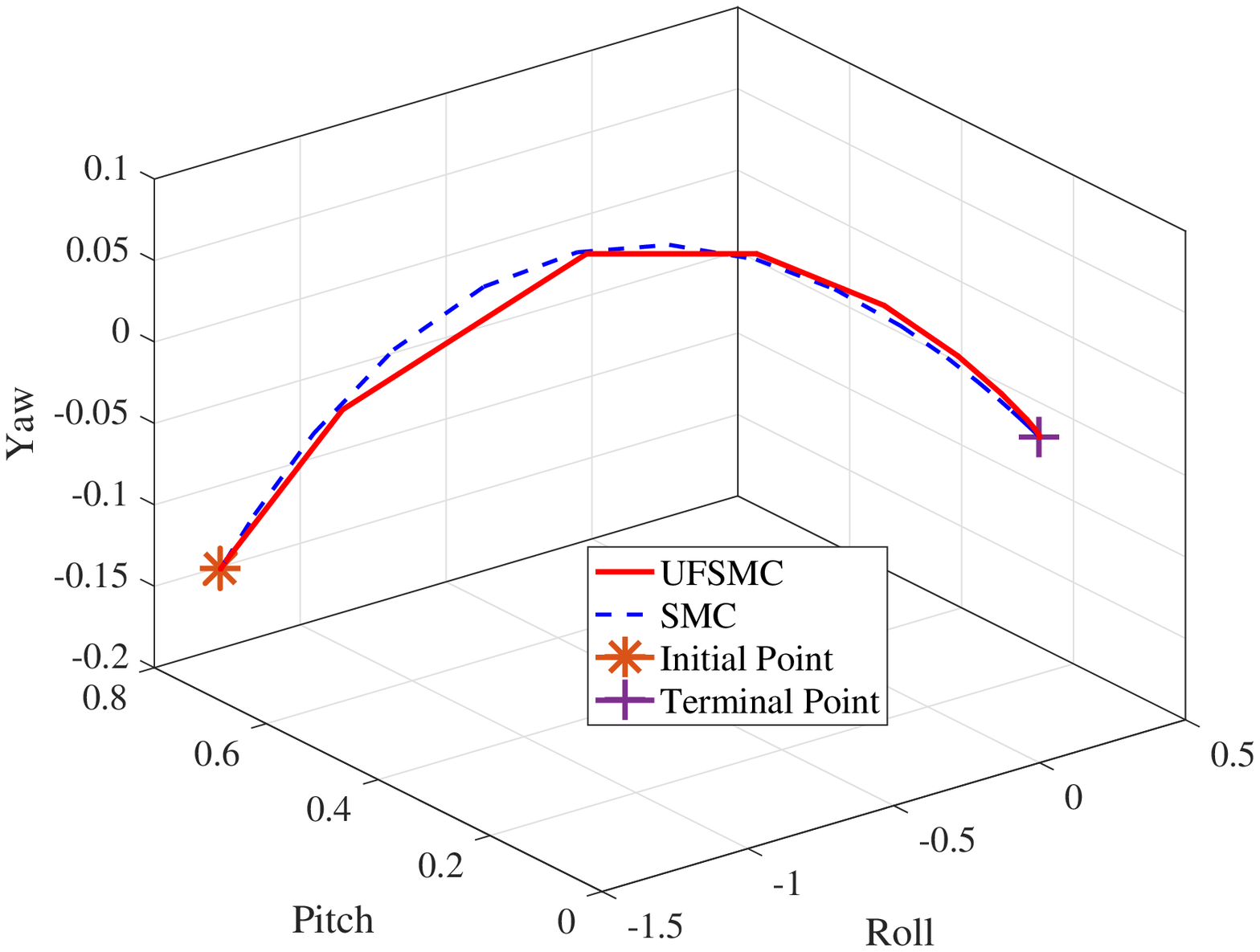}} 
\hspace{1cm}
\subfigure[Time response of the control torques]{
		\label{fig_D1}
		\includegraphics[width=3in]{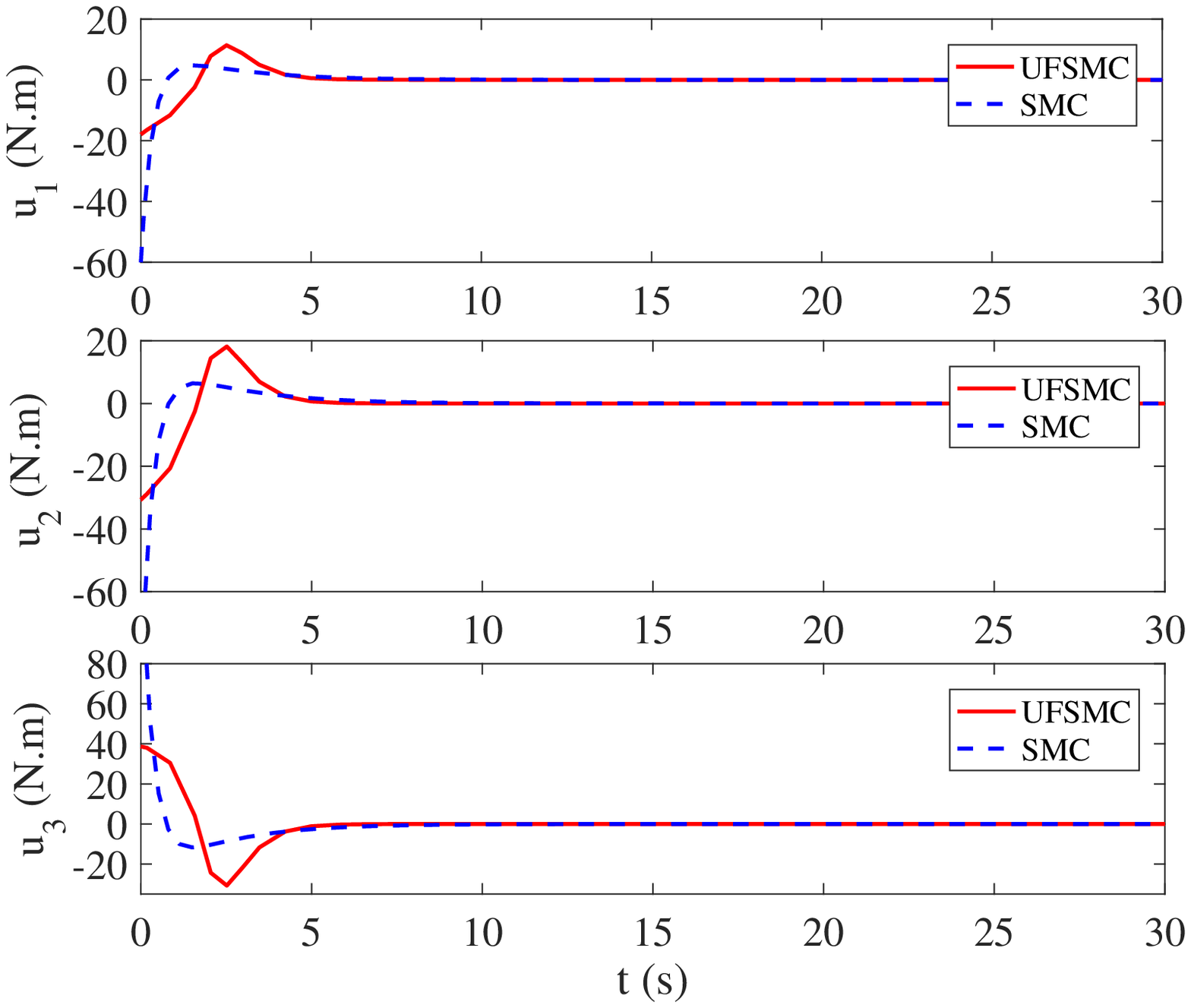}}
\caption{Comparison results of UFSMC law (\protect\ref{unwindingSMC}) and
SMC~\protect\cite{crassidis1996sliding} for Scenario A }
\label{fig_2}
\end{figure*}

\subsubsection{Simulation results for Scenario B}

The SMC controller in~\cite{crassidis1996sliding}, and the proposed UFSMC
law (\ref{unwindingSMC}) are adopted to do simulations for Scenario B. The
simulation results are summarized in Fig. \ref{fig_3}.

The response of the rotation angle $\theta\left(t\right)$ is shown in Fig. %
\ref{fig_A3}. The principle rotation angle $\theta\left(t\right)$ converges
to $0$ in about $14\mathrm{s}$ by adopting the SMC controller in~\cite%
{crassidis1996sliding}, while $\theta\left(t\right)$ converges to $2\pi$ in
about $6\mathrm{s}$ by adopting the proposed UFSMC~(\ref{unwindingSMC}).
This means that the rigid spacecraft needs to rotate $3.5036$ to reach the
desired attitude under the controller SMC in~\cite{crassidis1996sliding},
while the rigid spacecraft only needs to rotate $2.77$ to reach the desired
attitude under the proposed UFSMC~(\ref{unwindingSMC}). The behavior of
angular velocity error $\omega_{\mathrm{e}i}, i=1,2,3$ is shown in Fig. \ref%
{fig_C3}. It can be observed from Fig. \ref{fig_C3} that the attitude
velocity of the rigid spacecraft (\ref{systemmodel}) converges to $0$ in
about $6\mathrm{s}$ by using the proposed UFSMC law (\ref{unwindingSMC}),
while the SMC law needs a longer time. The spacecraft attitude responses
using Euler angles, i.e., Roll, Pitch, and Yaw, are the roll, pitch, and yaw
angles, respectively) are shown in Fig. \ref{fig_B3}. The maneuver angle of
the UFSMC law~(\ref{unwindingSMC}) is smaller than that of SMC law. This
means that the presented UFSMC law~(\ref{unwindingSMC}) can avoid the
unwinding phenomenon successfully, but the SMC controller can not. The
control torques $u_{i},i=1,2,3$ are shown in Fig. \ref{fig_D3}, which
indicates that the attitude maneuver is effectively settled by the UFSMC law(%
\ref{unwindingSMC}) and SMC controller. It can also be observed that the
control torque of the proposed control law is less than that of the SMC
controller.
\begin{figure*}[!t]
\centering
\subfigure[Time response of rotation angle $\theta\left(t\right)$]{
		\label{fig_A3}
		\includegraphics[width=3in]{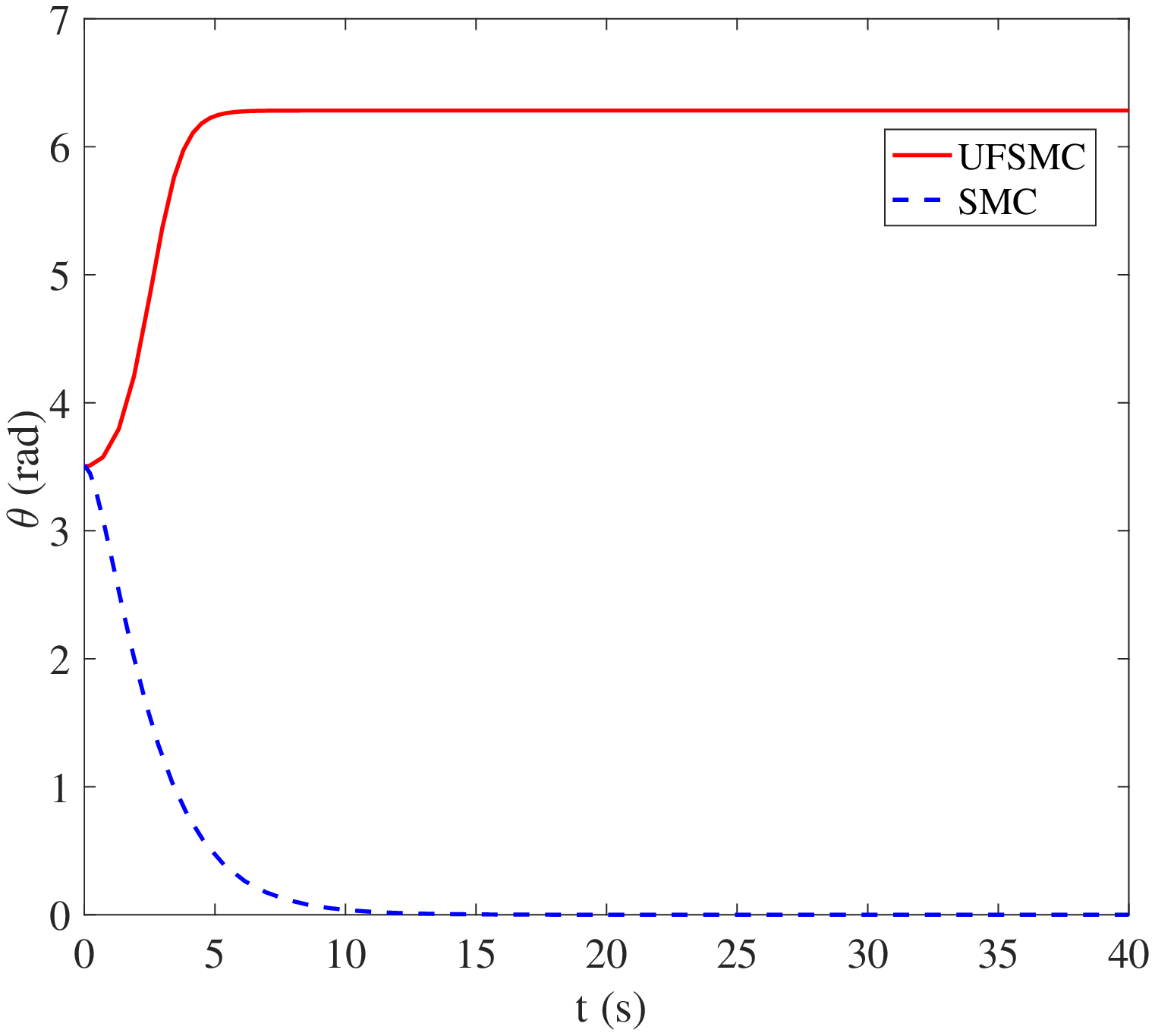}} 
\hspace{1cm}
\subfigure[Time response of the angular velocity $\boldsymbol{\omega}_{\mathrm{e}}$]{
		\label{fig_C3}
		\includegraphics[width=3in]{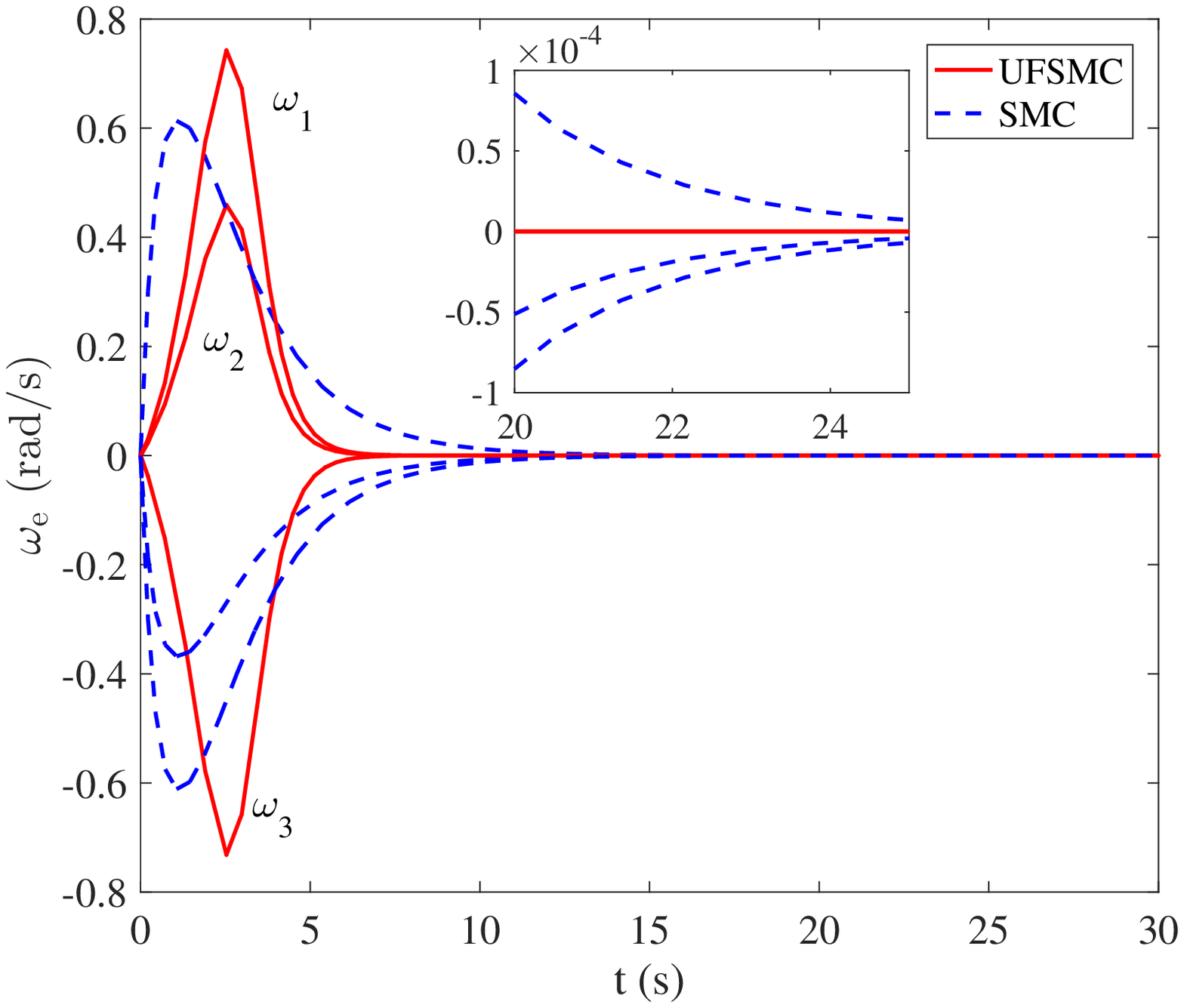}} 
\subfigure[Evolution of the Euler angles $\phi,\theta, \psi $ for the Scenario B]{
		\label{fig_B3}
		        \includegraphics[width=3in]{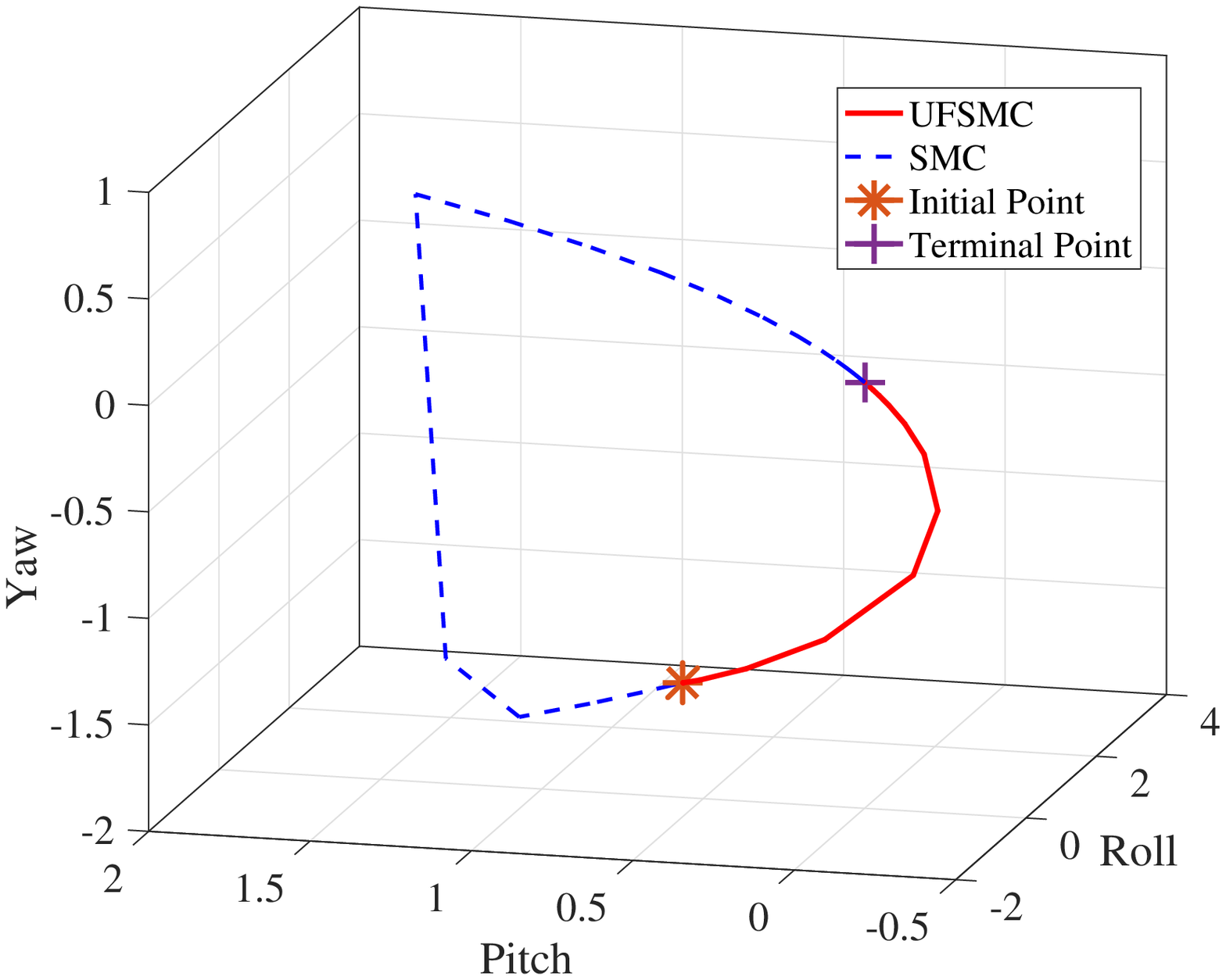}} 
\hspace{1cm}
\subfigure[Time response of the control torques $\boldsymbol{u}$]{\label{fig_D3}
		\includegraphics[width=3in]{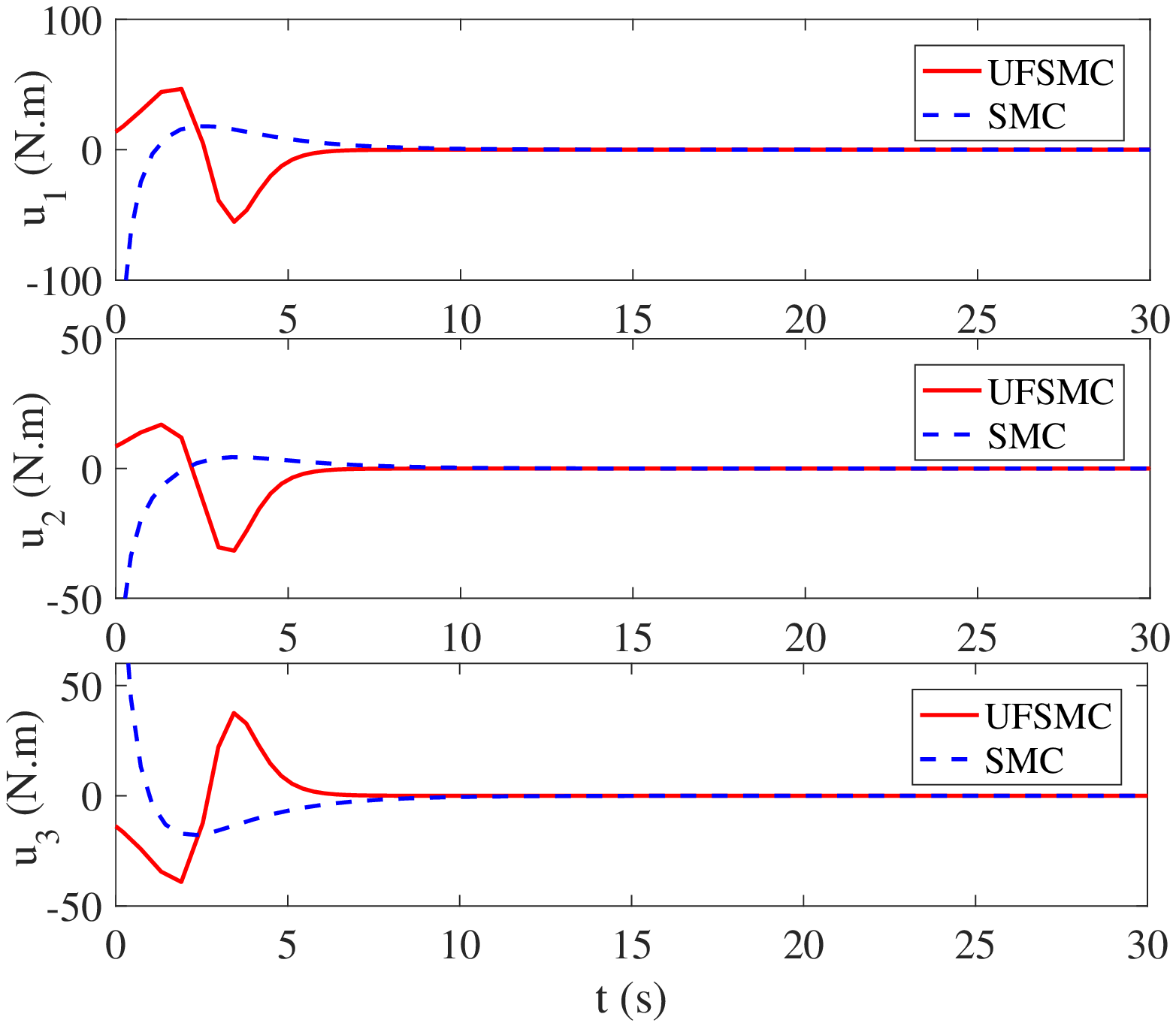}} 
\caption{Comparison results of UFSMC law (\protect\ref{unwindingSMC}) and
SMC~\protect\cite{crassidis1996sliding} and for Scenario B}
\label{fig_3}
\end{figure*}

In conclusion, the proposed UFSMC controller~(\ref{unwindingSMC}) satisfies
the control objective described in Section \ref{controlgoal}, and it
achieves higher pointing accuracy and better stability in a shorter time
compared with the SMC controller~(\ref{unwindingSMC}).%

\section{Conclusion}

In this paper, an unwinding-free sliding mode control law is presented for
the attitude maneuver control of a rigid spacecraft. By constructing a new
switching function, the unwinding-free property of the closed-loop attitude
maneuver control system of a rigid spacecraft is ensured when the system
states are on the sliding surface. Furthermore, by designing a sliding mode
control law with a dynamic parameter, the unwinding-free performance of the
closed-loop attitude maneuver control system of a rigid spacecraft is
guaranteed before the system states reach the sliding surface.
In addition, the switching function converges to zero in finite-time by the
developed control scheme. Finally, a numerical simulation is conducted to
demonstrate the effectiveness of the developed control law. The simulation
results show that the unwinding phenomenon is avoided by adopting the
designed switching surface and controller.

\bibliographystyle{ieeetr}
\bibliography{myreference}

\end{document}